\documentclass[
reprint,amsfonts, amssymb, amsmath,  showkeys,pra, superscriptaddress, twocolumn,longbibliography,nofootinbib]{revtex4-1}

\usepackage{graphicx}
\usepackage{array}
\usepackage{color}
\usepackage{float}
\usepackage{physics}
\usepackage{enumerate}
\usepackage{enumitem}
\usepackage{lipsum}
\usepackage{amsfonts}
\usepackage{mathtools}
\usepackage{dsfont}
\usepackage{booktabs}
\usepackage{url}
\usepackage{braket}

\usepackage{amsthm}
\usepackage{bm}
\usepackage[T1]{fontenc}

\usepackage[colorlinks=true,citecolor=magenta,linkcolor=magenta]{hyperref}

\usepackage{thm-restate}

\usepackage{tikz}
\usetikzlibrary{fadings}
\usepackage{mathrsfs}

\newtheorem{thm}{Theorem}

\usepackage{comment}

\newtheorem{lemma}[thm]{Lemma}

\theoremstyle{remark}

\definecolor{antonio}{rgb}{.2,.5,.1}

\newcommand{\mcl}{\mathcal{L}}
\newcommand{\mcc}{\mathcal{C}}
\newcommand{\mch}{\mathcal{H}}
\newcommand{\mco}{\mathcal{O}}
\newcommand{\mcu}{\mathcal{U}}
\newcommand{\mcm}{\mathcal{M}}

\newcommand{\mct}{\mathcal{T}}
\newcommand{\mcd}{\mathcal{D}}
\newcommand{\mcw}{\mathcal{W}}
\newcommand{\mbi}{\mathbb{I}}
\newcommand{\mbc}{\mathbb{C}}
\newcommand{\mbz}{\mathbb{Z}}
\newcommand{\mbe}{\mathbb{E}}
\newcommand{\mbs}{\mathbb{S}}

\newcommand{\mbu}{\mathbb{U}}

\newcommand{\id}{\mathds{1}}
\newcommand{\om}{\ketbra{\Omega}}
\newcommand{\lm}{\lambda}
\newcommand*{\defeq}{\mathrel{\vcenter{\baselineskip0.5ex \lineskiplimit0pt
                     \hbox{\scriptsize.}\hbox{\scriptsize.}}}%
                     =}

\begin{document}
\title{Real classical shadows} 

\author{Maxwell West}
\thanks{westm2@student.unimelb.edu.au}
\affiliation{School of Physics, University of Melbourne, Parkville, VIC 3010, Australia}
\affiliation{Theoretical Division, Los Alamos National Laboratory, Los Alamos, NM 87545, USA}

\author{Antonio Anna Mele}
\affiliation{Dahlem Center for Complex Quantum Systems, Freie Universität Berlin, 14195 Berlin, Germany}
\affiliation{Theoretical Division, Los Alamos National Laboratory, Los Alamos, NM 87545, USA}

\author{Mart\'{i}n Larocca}
\affiliation{Theoretical Division, Los Alamos National Laboratory, Los Alamos, NM 87545, USA}
\affiliation{Center for Non-Linear Studies, Los Alamos National Laboratory, 87545 NM, USA}

\author{M. Cerezo}
\thanks{cerezo@lanl.gov}
\affiliation{Information Sciences, Los Alamos National Laboratory, Los Alamos, NM 87545, USA}

\begin{abstract}
Efficiently learning expectation values of a quantum state using classical shadow tomography has become a fundamental task in quantum information theory. In a classical shadows protocol, one  measures a state in a chosen basis $\mathcal{W}$ after it has evolved under a unitary transformation randomly sampled from a chosen distribution $\mathcal{U}$. In this work we study the case where $\mathcal{U}$ corresponds to either local or global orthogonal Clifford gates, and $\mathcal{W}$ consists of real-valued vectors. Our results show that for various situations of interest, this ``real'' classical shadow protocol improves the sample complexity over the standard scheme based on general Clifford unitaries. For example, when one is interested in estimating the expectation values of  arbitrary real-valued observables, global orthogonal Cliffords typically decrease the required number of samples by a factor of two. More dramatically, for $k$-local observables composed only of real-valued Pauli operators,  sampling local orthogonal Cliffords leads to a reduction by an exponential-in-$k$ factor in the sample complexity over local unitary Cliffords. Finally, we show that by measuring in a basis containing  complex-valued vectors, orthogonal shadows can, in the limit of large system size, exactly reproduce the original unitary shadows protocol.
\end{abstract}

\maketitle

\section{Introduction}
Classical shadows has recently emerged as a powerful framework for learning the expectation values of a large number of operators from few copies of a state $\rho$ of interest~\cite{huang2020predicting,aaronson2018shadow,cotler2020quantum,sugiyama2013precision,elben2020mixed,rath2021quantum,zhao2021fermionic,low2022classical,van2022hardware,wan2023matchgate,ippoliti2024classical,zhao2024group,king2024triply,jerbi2023shadows,koh2022classical,chen2021robust,hearth2024efficient,bertoni2024shallow,chan2022algorithmic,low2022classical,sauvage2024classical,sack2022avoiding,wan2022matchgate,Koh_2022,Helsen_2023,kunjummen2023shadow,Wu_2024,Zhao_2024,bermejo2024quantum,angrisani2024classically,west2024random,grier2024sample,brandao2020fast,bertoni2022shallow,somma2024shadow}. As technological progress allows us to create quantum systems that are far too large for full state tomography,  such randomized measurement techniques are increasingly becoming the norm for unknown state characterization. Instead of attempting to fully characterize a state, the remarkable promise of classical shadows is that one can  perform an extremely small number of measurements of $\rho$, and later use them to produce estimates $\hat{o}_i\approx\Tr[\rho O_i]$ for some collection of observables  
 $\{O_i\}_{i=1}^M$ that was not fixed a priori. Such estimates can be obtained to within additive error $\varepsilon$ (with high probability) by measuring only
\begin{equation}\label{eq:sc}
S=\mco\left(\frac{\log M}{\varepsilon^2}\max_i {\rm Var}_{\hspace{0.5mm}\mcu}[\hat{o}_i ]\right)
\end{equation}
copies of $\rho$~\cite{huang2020predicting}. 
Importantly, the variances of the shadow estimators -- and thus the sample complexity -- depend on the distribution $\mcu$ used in the randomized measurements. Hence, different  choices of  $\mcu$  lead to different families of observables with small variances which can  be sample-efficiently estimated. 

The search for distributions of unitaries that display favorable variance scaling for operators of general interest is thus an active area of research, with bespoke choices now known (for example) for operators which are local~\cite{huang2020predicting}, constitute a Majorana monomial of low degree~\cite{zhao2021fermionic,wan2022matchgate}, or are permutation symmetric~\cite{sauvage2024classical}. 
Conversely, one can pick a distribution of unitaries and determine the operators leading to an efficient shadow protocol. This has lead to the characterization of a growing but still relatively short list of distributions, including random local and global Clifford gates~\cite{huang2020predicting}, fermionic Gaussian unitaries~\cite{zhao2021fermionic,wan2022matchgate,denzler2023learningfermioniccorrelationsevolving} and symplectic unitaries~\cite{west2024random}.  

In this work we contribute to the ever-growing literature of randomized measurement tomography~\cite{Elben_2022}. Specifically,  we analyze classical shadow schemes based on sampling unitaries from the Haar measure over the local and global orthogonal groups.   As we shall see (and analogously to the unitary~\cite{zhu2017multiqubit} and fermionic Gaussian unitary~\cite{wan2022matchgate} cases) the corresponding classical shadow protocols are identical to those obtained by respectively taking local and global  \textit{real Cliffords}, by virtue of their forming an orthogonal 3-design~\cite{hashagen2018real}. 
In the case of global real Cliffords and a measurement basis consisting of real vectors, we find that the class of operators that can be efficiently learnt are the symmetric (equivalently, real-valued) observables with 2-norm polynomial in the system size. Moreover, the variance of the resulting estimators is typically (in a to-be-discussed sense)  reduced by a factor of two compared to the case of global unitary Cliffords~\cite{huang2020predicting} (which corresponds to the best suited scheme to learning observables with small 2-norm). In the case of local real Cliffords and a real measurement basis we will see that the efficiently learnable class of observables are the local operators containing in their Pauli decompositions only strings consisting of $I$, $X$  and $Z$. For such an operator of support $k$, we find an improvement by an exponential factor in the variance, from $4^k$ to $3^k$, when using local real Cliffords over local unitary Cliffords. When we restrict further to a single such Pauli string the respective variance bounds become $3^k$ and $2^k$, another exponential factor improvement.  As  schematically depicted in Fig.~\ref{fig:1}, these variance improvements come at the cost of restricting to a strict subset of all observables. However, when  this set contains the observables of interest, orthogonal shadows yield a clear improvement. 

More generally, when measuring in a complex basis (i.e., a basis containing at least one vector $\ket{w}$ which is linearly independent of $\ket{w^*}$), we find that the induced shadow protocol is a linear combination of the previously described real protocols and that of unitary shadows. Remarkably, we further find a simple measurement basis for which, in the limit of large Hilbert space dimension, orthogonal shadows exactly reproduces the unitary shadows protocol of Ref.~\cite{huang2020predicting}. Indeed, we show that by tuning the \textit{reality} of the basis, one can interpolate between real and unitary shadows, allowing both for the targeting of a continuum of distinct visible spaces and bias-variance trade-offs~\cite{van2022hardware}.

\section{Background}

For completeness, we begin by briefly describing the core concepts of classical shadows. For a more thorough introduction to this topic, we refer the reader to Refs.~\cite{huang2020predicting,elben2022randomized}. Let $\rho$ be an unknown $n$-qubit quantum state, i.e., a trace-one positive-semidefinite linear operator on the $d$-dimensional vector space $\mch$, where $d=2^n$. We will use $\mcl(\mch)$ to denote the space of linear operators on $\mch$. A classical shadows protocol is defined by a choice of distribution $\mcu$ (i.e., a set of unitaries equipped with a probability measure) and a measurement basis $\mcw=\{\ket{w}\}_w$. Copies of the unknown state $\rho$ are evolved by a  randomly sampled $U\sim\mcu$ and then measured in the basis $\mcw$, producing pairs $(U,\ket{w})$. Upon obtaining such a pair, the state $U^\dagger\Pi_wU$, where $\Pi_w \defeq \ketbra{w}$, is stored. This  procedure can be seen to implement the quantum channel
\begin{align}
\mcm_{\mcu,\mcw}(\rho)&=  \sum_{w\in\mcw} \int_{U\sim\hspace{0.5mm}\mcu}  \Tr\left[\rho U^\dagger\Pi_wU\right] U^\dagger\Pi_w U \label{eq:mc1}\\
&=\sum_{w\in\mcw}  \Tr_1\left[\left(\rho \otimes \id\right) \int_{U\sim\hspace{0.5mm}\mcu} U^{\dagger\otimes 2}\Pi_w^{\otimes 2}U^{\otimes 2}\right]\label{eq:mc2}\,.
\end{align}
Here, $\int_{U\sim\hspace{0.5mm}\mcu}$ denotes the average over the distribution $\mathcal{U}$. 
For each pair $(U,{w})$, one defines the \textit{classical shadows} $\hat\rho$ of $\rho$ as
\begin{equation}
\hat{\rho} = \mcm_{\mcu,\mcw}^{-1}\left(U^\dagger\Pi_wU\right),
\end{equation}
where  $\mcm_{\mcu,\mcw}^{-1}$  denotes the pseudo-inverse of $\mcm_{\mcu,\mcw}$. The visible space of the shadow protocol, i.e., the image~\cite{van2022hardware} of $\mcm_{\mcu,\mcw}$, is defined as
\begin{equation}\label{eq:vis}
\mathsf{VisibleSpace}\left(\mcu, \mcw\right) = {\rm span }\left\{ U^\dagger \Pi_{w} U \right\}_{U\sim\hspace{0.5mm}\mcu, w \in\mcw}\,,
\end{equation}
and in general is a strict subspace of $\mcl(\mch)$.

\begin{figure}[t!]
\includegraphics[width=\linewidth]{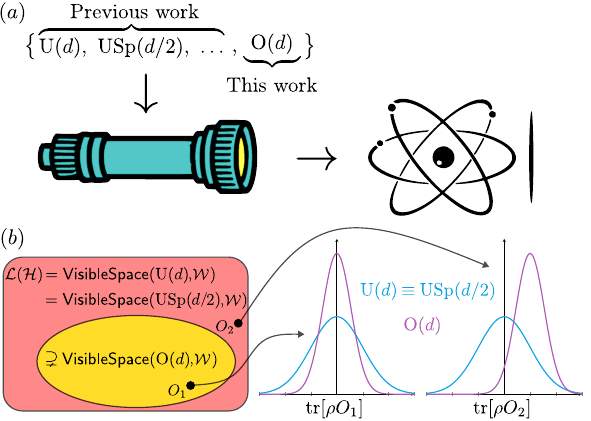}
\caption{(a) Schematic depiction of the classical shadows procedure. An unknown quantum state is evolved with unitaries sampled from a distribution $\mathcal{U}$, and measured in a basis $\mathcal{W}$, thus ``illuminating'' different aspects of the state. (b) The resulting \textit{classical shadows} can be used to estimate expectation values with respect to the state. The $\mathcal{U}$-dependent set of observables for which these estimators are unbiased form the \textit{visible space} of the protocol. Taking $\mathcal{U}$ to be Haar random over the unitary ${\rm U}(d)$ or symplectic ${\rm USp}(d/2)$~\cite{west2024random} group, and $\mathcal{W}$ the computational basis, leads to a tomographically complete protocol (i.e., with a maximal visible space). For an operator $O_1$ ($O_2$) which is in (not in) the orthogonal visible space, the estimator is unbiased (biased). However, orthogonal shadows lead to smaller variances for $O_1$, and therefore an improved sample complexity.  }
\label{fig:1}
\end{figure}

Given an observable $O$, one defines its estimator  in terms of the classical shadows as
\begin{equation}
\hat{o}=\Tr[O\hat{\rho}]\,,
\end{equation}
which can be shown to be \textit{unbiased} if  $O\in\mathsf{VisibleSpace}\left(\mcu, \mcw\right)$~\cite{van2022hardware}. Finally, we wish to know the variance of the estimators, which will allow us to determine the number of copies of $\rho$ needed to obtain reliable estimates (see Eq.~\eqref{eq:sc}). As such, one needs to evaluate the quantity 
\begin{equation}
{\rm Var}[\hat{o}_i ]=\mbe[\hat{o}_i^2]-\mbe[\hat{o}_i]^2\,,
\end{equation}
where
\begin{align}
\mbe[\hat{o}^2] &= \sum_w\int_{U\sim\hspace{0.5mm}\mcu}   \Tr\left[\rho U^\dagger\Pi_wU\right]\Tr\left[O \mcm_{\mcu,\mcw}^{-1}\left(U^\dagger\Pi_wU\right)\right] ^2 \nonumber \\
&=\sum_w  \Tr\bigg[\left(\rho\otimes \mcm_{\mcu,\mcw}^{-1}(O)\otimes  \mcm_{\mcu,\mcw}^{-1}(O)\right) \times\nonumber \\
&\hspace{20mm}\int_{U\sim\hspace{0.5mm}\mcu}  U^{\dagger\otimes 3}\Pi_w^{\otimes 3}U^{\otimes 3} \bigg]\,,\label{eq:ea2}
\end{align}
and
\small
\begin{align}
\mbe[\hat{o}]^2 &=\left( \sum_w\int_{U\sim\hspace{0.5mm}\mcu}   \Tr\left[\rho U^\dagger\Pi_wU\right]\Tr\left[O\mcm_{\mcu,\mcw}^{-1}\left(U^\dagger\Pi_wU\right)\right]\right) ^2 \nonumber \\
&= \left(\sum_w  \Tr\left[\left(\rho\otimes  \mcm_{\mcu,\mcw}^{-1}(O)\right)\int_{U\sim\hspace{0.5mm}\mcu}  U^{\dagger\otimes 2}\Pi_w^{\otimes 2}U^{\otimes 2} \right] \right) ^2\!\!. \label{eq:e2a}
\end{align}
\normalsize
These expressions are derived and discussed in more detail in the Appendix.

The key takeaway is that the properties of a given shadow protocol are fully determined by Eqs.~\eqref{eq:mc2},~\eqref{eq:ea2} and~\eqref{eq:e2a}. Indeed, one can compute the variance, and thus the sample complexity of a given scheme by knowing the inverse of the channel $\mcm_{\mcu,\mcw}$, and by calculating the values of the integrals
\begin{equation}\label{eq:moments}
\mct^{(k)}_\mcu (\Pi_w^{\otimes k}) = \int_{U\sim\hspace{0.5mm}\mcu} U^{\otimes k} \Pi_w^{\otimes k} U^{\dagger \otimes k}\,,
\end{equation}
for $k=2,3$. If $\mathcal{U}$ is the uniform (Haar) measure over a group, one can evaluate Eq.~\eqref{eq:moments} via the \textit{Weingarten calculus}~\cite{weingarten1978asymptotic,mele2023introduction}, the  relevant details of which are discussed in the Appendix. An important corollary of the fact that only these second and third order moments of the distributions are relevant is that any  $\mcu'$ which is a state 3-design with respect to $\mcu$ will produce identical classical shadows. This is exploited in Ref.~\cite{huang2020predicting} to use the (easier to prepare) Cliffords in lieu of Haar random unitaries, and in Ref.~\cite{wan2022matchgate} to use ``matchgate Cliffords'' instead of arbitrary matchgate circuits.

\section{Results}

Equipped with the general framework of classical shadows presented in the previous section, we now proceed to present our main results. In particular, we will focus on the cases when $\mathcal{U}$ corresponds to the Haar measure over the global orthogonal group ${\rm O}(d)$, and over the product of local orthogonal groups   ${\rm O}(2)^{\otimes n}$. For all practical purposes, we will assume that these shadow protocols are respectively implemented by sampling\footnote{This can be done efficiently, see e.g. Ref.~\cite{hashagen2018real}.} real Clifford unitaries from the uniform distribution over the discrete groups ${\rm Cl}_n\cap {\rm O}(d)$ or $({\rm Cl}_1\cap {\rm O}(2))^{\otimes n}$, where here ${\rm Cl}_n$ denotes the $n$-qubit Clifford group  (i.e. the normalizer of the Pauli group~\cite{calderbank1997quantum}). Indeed, since real Cliffords form an orthogonal 3-design~\cite{nebe2001invariants,nebe2006self,hashagen2018real}, these will lead to easier to implement protocols, while preserving all the relevant statistical properties. 
Then, throughout most of this work we will focus on the case where the measurement basis is real, i.e.\footnote{We will   always take  transposes  with respect to the computational basis.}, $\Pi_w=\Pi_w^\mathsf{T},\, \forall w\in\mathcal{W}$. We shall term any such choice, combined with global or local (Clifford) unitaries,  a \textit{real classical shadows protocol}. The more general case of a complex-valued measurement basis will be discussed at the end of this section.

We begin with a characterization of the measurement channel resulting from taking a real basis and the distribution of global orthogonal operators.
\begin{restatable}{prop}{gmc}\label{prop:gmc}
The measurement channel resulting from real classical shadows with global orthogonal evolution and any real measurement basis $\mcw$ is given by
\small
\begin{equation}\label{eq:gmc}
  \mcm_{{\rm O}(d),\mcw}(A)  = \frac{\Tr[A]\mathds{1}+A+ A^\mathsf{T} }{d+2} = \mcd_{d/(d+2)}(A_{\rm sym})\,,
  \end{equation}
  \normalsize
where 
\begin{equation}
A_{\rm sym}=\frac{A+A^\mathsf{T}}{2}\,,   
\end{equation}
is the symmetric component of $A$ and
\begin{equation}\label{eq:depol}
  \mcd_p(A) = \frac{p\Tr[A]}{d}\id + (1-p)A  
\end{equation}
is the depolarizing channel of strength $p$.  
\end{restatable}
We refer the reader to the Appendix for a proof of this proposition, as well as that of all our other results. 

Notably, we can compare the result in Proposition~\ref{prop:gmc} to the channel obtained by taking $\mathcal{U}$ to be the Haar measure over ${\rm U}(d)$, which acts  as~\cite{huang2020predicting}
\begin{equation}\label{eq:gumc}
  \mcm_{{\rm U}(d),\mcw}(A)  =  \mcd_{d/(d+1)}(A)\,.
\end{equation}
Apart from a small change in the strength of the depolarization, the difference between Eqs.~\eqref{eq:gmc} and~\eqref{eq:gumc} is that in the orthogonal case only the symmetric component of the input operator is acted upon\footnote{More technically, one can view the channel as the concatenation $\mcm_{{\rm O}(d),\mcw} = \mcd_{d/(d+2)}\circ\mct_{\mbz_2}$, where $\mct_{\mbz_2}$ denotes the twirl over a representation $S$ of $\mbz_2$ on $\mcl(\mathcal{H})$ that maps $A\mapsto A^\mathsf{T}$. A projector onto its image, this twirl annihilates the antisymmetric components of $A$.}.
Hence,  Proposition~\ref{prop:gmc} immediately implies that orthogonal shadows cannot yield unbiased estimates of expectation values of operators $A$ which do not satisfy $A=A^\mathsf{T}$, as the asymmetric component $(A-A^\mathsf{T})/2$ of $A$ is  annihilated by the channel, and thus lives outside its visible space (Eq.~\eqref{eq:vis}). That is, 
\begin{align}\label{eq:vis-o}
\mathsf{VisibleSpace}\left({\rm O}(d), \mcw\right) &= \{A\in\mathcal{L}(\mathcal{H})\,|\, A=A^\mathsf{T}\}\nonumber\\
&=\mathbb{C}\mathfrak{o}(d)\,,
\end{align}
where $\mathbb{C}\mathfrak{o}(d)$ denotes the complexification of the orthogonal algebra $\mathfrak{o}(d)$.  As such, we find 
\begin{equation*}
\frac{{\rm dim }(\mathsf{VisibleSpace}\left({\rm O}(d), \mcw\right)) }{{\rm dim }(\mcl)}=\frac{d(d-1)/2}{d^2}\to 1/2\,,
\end{equation*} 
which shows that the dimension of the visible space is roughly half of that of the total operator space. 

For operators within the visible space of Eq.~\eqref{eq:vis-o}, we can obtain a modest advantage over global unitaries,  as quantified in the following result:
\begin{restatable}{prop}{gv}\label{prop:gvar}
The variance of the estimators of real classical shadows with global orthogonal  evolution is given by
\small
\begin{align*}
  {\rm Var}(\hat{a}) &=\frac{d+2}{2d+8} \left(\Tr[A_{{\rm sym}; 0}^2] + 4\Tr[\rho A_{{\rm sym}; 0}^2]\right)- \Tr[A_{{\rm sym}; 0} \rho]^2 \,,
  \end{align*} 
  \normalsize
  where 
  \begin{equation*}
      A_{{\rm sym}; 0} = \frac{A+A^\mathsf{T}}{2} - \frac{\Tr[A]}{d}\id
  \end{equation*}
  is the traceless component of the symmetric component of $A$.
\end{restatable}
To see the potential for an  advantage over global unitaries we note that by Holder's inequality $\Tr[\rho A_{{\rm sym}; 0}^2]\leq \|\rho\|_1\|A_{{\rm sym}; 0}^2\|_\infty=\|A_{{\rm sym}; 0}\|_\infty^2$ will typically\footnote{This is explored numerically in Fig.~\ref{fig:numerics}, and discussed further in Appendix~\ref{sec:appc}.} be much smaller than $\Tr[A_{{\rm sym}; 0}^2] =\|A_{{\rm sym}; 0}\|_2^2 $. Hence,  in the limit of large $d$, Eq.~\eqref{eq:gmc} becomes 
\begin{equation}
{\rm Var}(\hat{a})\sim \frac{1}{2}\|A_{{\rm sym}; 0}\|_2^2\,,
\end{equation}
which constitutes an improvement by a factor of two over the result of global unitaries~\cite{huang2020predicting}. This result is consistent with the characterization (for symmetric operators) of $\mcm_{{\rm O}(d),\mcw}$ and $\mcm_{{\rm U}(d),\mcw}$ as being depolarizing channels, with the depolarization strength of $\mcm_{{\rm U}(d),\mcw}$ being higher. Intuitively, faced with a higher probability of the maximally mixed state being outputted, more shadows will be needed for global unitaries.  
From this result, we conclude that for the task of estimating expectation values of real observables with small $2$-norm, real classical shadows outperform  by a constant factor the standard unitary classical shadows scheme. An example of a real observable with small $2$-norm is given by real states: computing their expectation values (overlap) with respect to an unknown target state might be useful for verification and benchmarking purposes.\\

Mirroring the developments in Ref.~\cite{huang2020predicting} we next turn to the case of \textit{local} orthogonal operators, i.e., random single  orthogonal unitaries applied to each qubit independently. Here we find that the following proposition holds.

\begin{figure}
    \centering
    \includegraphics[width=\linewidth]{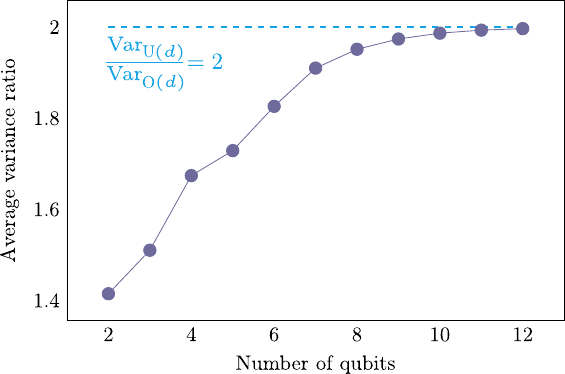}
    \caption{The ratio of the estimator variances in the global real and global unitary Clifford protocols, for a random initial state and random (symmetric) observables, averaged over 500 instances. The observable distribution is given by sampling matrices with real and imaginary components entry-wise uniformly sampled in $[-1, 1]$, then adding the adjoint and normalizing by the Schatten 2-norm. We note that different distributions may lead to  different behaviour of the ratio, but that nonetheless we will always have ${\rm Var}_{{\rm O}(d)}\leq {\rm Var}_{{\rm U}(d)}$ (see Appendix~\ref{sec:appc}). }
    \label{fig:numerics}
\end{figure}

\begin{restatable}{prop}{lmc}\label{prop:lmc}
The measurement channel resulting from real classical shadows with local orthogonal evolution is given by
\begin{equation}
  \mcm_{{\rm O}(2)^{\otimes n},\mcc}\left(A\right)=\mcd_{1/2}^{\otimes n}\big(A_{\rm L.S.}\big)\,.
\end{equation}
Here $\mcc$ is the computational basis  and $A_{\rm L.S.}$ is the projection of $A$ onto the subspace of ``locally symmetric''  operators. That is,   if
  \begin{equation}
  A= \sum_{i_1,\ldots,i_n\in\{I,X,Y,Z\}}a_{i_1,\ldots,i_n}A_{i_1}\otimes\ldots\otimes A_{i_n}\,,
  \end{equation}
  then
    \begin{equation}
    A_{\rm L.S.} = \sum_{i_1,\ldots,i_n\in\{I,X,Z\}}a_{i_1,\ldots,i_n}A_{i_1}\otimes\ldots\otimes A_{i_n}.
  \end{equation}
The classical shadows are of the form 
  \begin{equation}
    \hat{\rho} = \bigotimes_{j=1}^n \left(2U_j^\dagger |\hat{b}_j\rangle\langle\hat{b}_j| U_j - \frac{1}{2}\id\right).
  \end{equation}
\end{restatable}
It is again interesting to compare the results of Proposition~\ref{prop:lmc} to the channel obtained by acting with the local unitaries, which takes the form~\cite{huang2020predicting}
\begin{equation}
    \mcm_{{\rm U}(2)^{\otimes n},\mcc}(A) = \mcd_{2/3}^{\otimes n}(A).
\end{equation}
The contrast between the two local channels is qualitatively analogous to that in the global case, with the orthogonal channel acting as a weaker depolarization channel on a smaller visible space. The effect here is much stronger, however, with the relative difference in the depolarization strength being much higher. This is matched by a corresponding decrease in the ratio of the dimension of the visible space of the local orthogonal channels to that of the full operator space. Indeed, one finds that 
\begin{equation*}
\frac{{\rm dim }(\mathsf{VisibleSpace}\left({\rm O}(2)^{\otimes n}, \mcc\right) )}{{\rm dim }(\mcl)}=(3/4)^n\to 0\,,
\end{equation*}
where
\begin{align}
    \mathsf{VisibleSpace}\left({\rm O}(2)^{\otimes n}, \mcc\right)&={\rm span}_{\mathbb{C}}\{\id,X,Z\}^{\otimes n}\nonumber\\
    &=\bigotimes_{j=1}^n \mathbb{C}\mathfrak{o}(2)\,.
\end{align}
Above, each $\mathbb{C}\mathfrak{o}(2)$ denotes the complexification of the local orthogonal algebra on the $j$-th qubit. 

If -- despite its smallness -- the observables in which one is interested  live within $\mathsf{VisibleSpace}\left({\rm O}(2)^{\otimes n}, \mcc\right)$, dramatic improvements in sample-complexity are possible. For example we have:
\begin{restatable}{crllr}{lp}\label{prop:lp}
Under the real local orthogonal channel, the variance of a weight-$k$ Pauli string containing only $I,\ X$ and $Z$ is upper bounded by $2^k$. 
\end{restatable}
\noindent
Here \textit{weight} denotes the number of qubits which are acted upon non-trivially.
In fact, one can see that the bound of $2^k$  is  \emph{optimal}. Specifically, assuming $k \le O(\log(n))$, any (possibly adaptive) single-copy protocol that estimates $\Tr(P \rho)$ to precision $\varepsilon$  with constant success probability for all $k$-local, $n$-qubit Pauli strings containing only $I$, $X$, and $Z$ requires at least $\Omega\left(\frac{2^k}{\varepsilon^2}\right)$ copies of $\rho$. This result follows by adapting the proof strategy from Theorem 26 in Ref.~\cite{king2024triply}, where in their ternary tree construction, we restrict to Pauli observables containing only $X$ and $Z$. \\

In a similar spirit, one can also show: 
\begin{restatable}{crllr}{lo}\label{prop:lo}
Under the real local orthogonal channel, the variance of a $k$-local operator which decomposes into Pauli strings containing only $I,\ X$ and $Z$  is upper bounded by  $3^k$. 
\end{restatable}
In both Corollary~\ref{prop:lp} and~\ref{prop:lo} we find an exponential (with respect to $k$) improvement in variance (and hence, sample-complexity) compared to what can be achieved when sampling local unitaries (namely, $3^k$ for individual Paulis and $4^k$ for general local operators). In fact, in all four of these cases we find that the variance is proportional to the dimension of the image under the action of the sampled unitaries on the subspace of operators in which the target observables live. In the case of local Paulis this result is consistent with the previous observation~\cite{king2024triply} (in the unitary case) of the sample complexity being inversely proportional to the probability of randomly sampling an operator which diagonalizes a given target Pauli string.

\begin{figure}
    \centering
    \includegraphics[width=\linewidth]{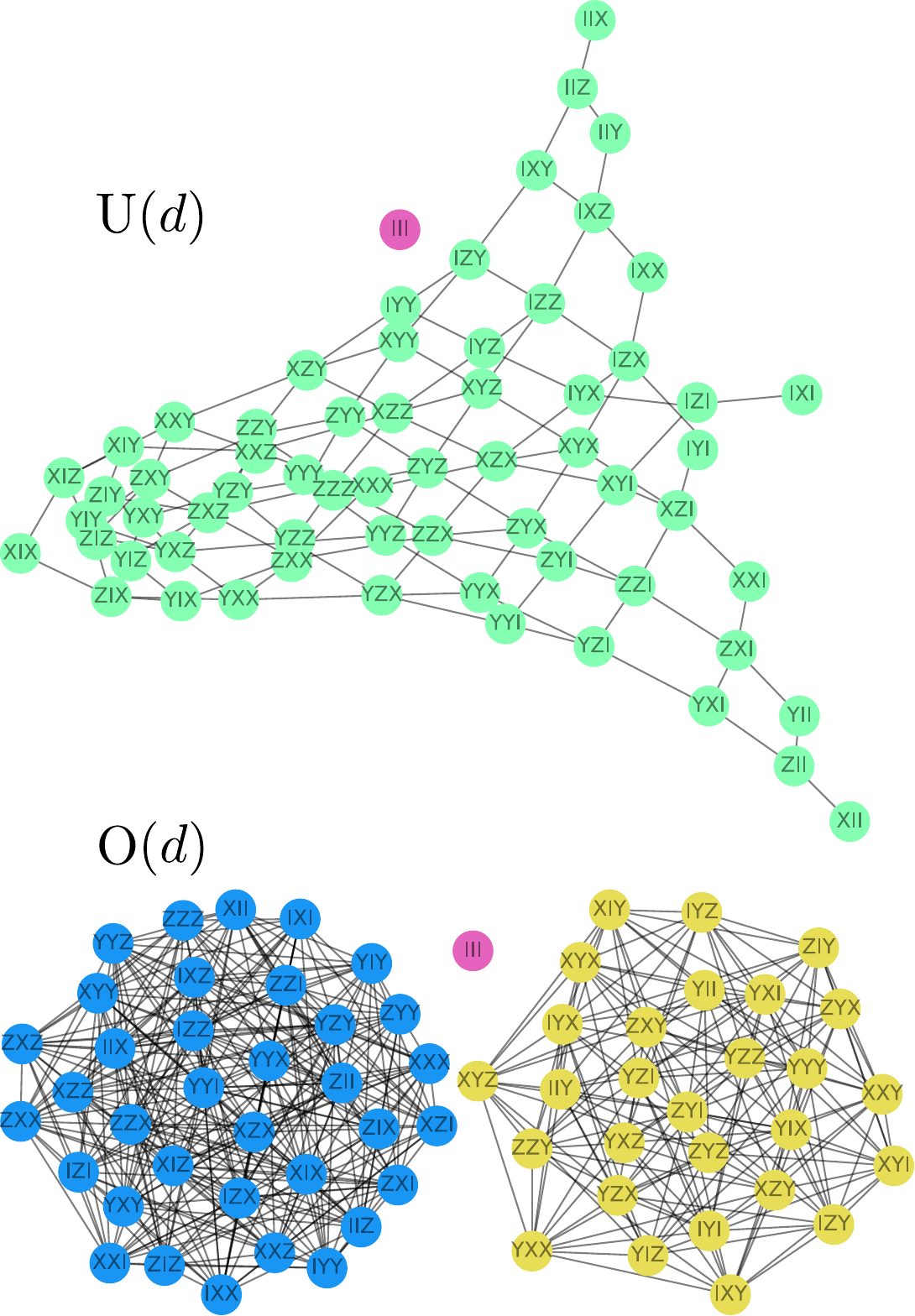}
    \caption{The ${\rm U}(d)$- and ${\rm O}(d)$-module structure of ${\rm End}\hspace{0.5mm}\mch$, as visualized by \textit{commutator graphs}~\cite{diaz2023showcasing,west2025graph}. For our purposes the important information is that each connected component of the above graphs corresponds to a different irreducible module of the adjoint action of the respective group. When $\ket{w}$ is a computational basis state, $\ketbra{w}\in{\rm span}_{\mathbb{R}} \{I,Z\}^{\otimes n}$, and $O\ketbra{w}O^{\mathsf{T}}$ evolves entirely within the blue module. On the other hand, $SH\ketbra{w}(SH)^{\dagger}\in{\rm span}_{\mathbb{R}} \{I,Y\}^{\otimes n}$, and its evolution therefore splits across the blue and gold modules, more closely approximating the unitary case (evolution across the green  module). }
    \label{fig:modules}
\end{figure}

Finally, we relax the requirement that the measurement basis be real. For each basis vector $\ket{w}\in\mcw$ we define a measure of its ``reality'', $\alpha_w=\abs{\braket{w|w^*}}^2$, along with a notion of the reality of the entire basis, $\alpha=\sum_w \alpha_w$\footnote{In the notation of the previously introduced twirl $\mct_{\mbz_2}(A)=(A+A^\mathsf{T})/2$ over $\mbz_2$, one has $\braket{\Pi_w|\mct_{\mbz_2}(\Pi_w)}_{\rm H.S.}=(1+\alpha_w)/2$.}. The previous quantity turns out to be  key for characterizing the ensuing shadow protocol as we find:
\begin{restatable}{prop}{cmc}\label{prop:cmc}
When measuring in a basis $\mcw_\alpha$ with reality $\alpha$, the global orthogonal measurement channel is given by
\begin{equation}\label{eq:cmc}
  \mcm_{{\rm O}(d),\mcw_\alpha}(A)  =  \mcd_{p_\alpha}(\widetilde{A}_\alpha)\,,
  \end{equation}
where $\mcd_{p_\alpha}$ is the
depolarizing channel of strength 
\begin{equation}
p_\alpha=\frac{d^2-\alpha}{(d-1)(d+2)}\,,
\end{equation}
and
\begin{equation}
\widetilde{A}_\alpha=\frac{(d^2-\alpha)A +(\alpha d+\alpha-2d)A^\mathsf{T} }{d(d-2+\alpha)}.
\end{equation} 
\end{restatable}
When the basis is comprised entirely of real vectors  ($\alpha=d$) Proposition~\ref{prop:cmc} reduces to Proposition~\ref{prop:gmc}.
However, as we discuss in the Appendix, in the limit $d\to\infty$ one has
\begin{equation}
\widetilde{A}_{fd}\to \frac{A +f A^\mathsf{T} }{f+1}\,,
\end{equation}
with $f=\alpha/d ,\ 0\leq f \leq 1$. Comparing the previous equation with the result of Proposition~\ref{prop:gmc}, and recalling that unitary shadows acts as a depolarizing channel on $A$, we see that by measuring in  bases of decreasing
reality we can interpolate between the case of real and unitary shadows, 
while still evolving the unknown states via the orthogonal group. For example, as we show in the Appendix, measuring in a (predetermined, protocol-defining) random basis will yield $f\approx 2/(d+1)$, approximately reproducing unitary shadows. Alternately, measuring in the basis given by $\{S^{\otimes n}H^{\otimes n}\ket{z}\}_{z=0}^{d-1}$ with $S$ the phase gate, $H$ the Hadamard and $\{\ket{z}\}_{z=0}^{d-1}$ the computational basis, leads to $f=0$. In the limit $d\to\infty$ we can therefore  exactly reproduce unitary shadows through orthogonal evolution and a readily implementable measurement. This can be seen as an extension of the recent result that there exist states that, when twirled by the orthogonal group, produce a unitary state 3-design~\cite{schatzki2024random}.
We can gain some intuition for this result by examining the module structure of ${\rm End}\hspace{0.5mm}\mch$ with respect to both the unitary and orthogonal groups. One can show (see Fig.~\ref{fig:modules}) that under the adjoint action of ${\rm U}(d)$, ${\rm End}\hspace{0.5mm}\mch$ breaks up into a trivial irrep (spanned by the identity), and a ($d^2-1$)-dimensional irrep (the remainder of the space). On the other hand, under   the adjoint action of ${\rm O}(d)$, one finds three irreps: the identity, the symmetric traceless operators, and the anti-symmetric operators. 
In the qubit case, these latter two spaces are spanned respectively by the non-identity Pauli strings with an even (odd) number of $Y$s. For any computational basis state $\ketbra{w}\in {\rm span}_{\mathbb{R}} \{I,Z\}^{\otimes n}$, therefore, we find that orthogonal evolution takes $\ketbra{w}$ only within the even-$Y$ component (the blue irrep in Fig.~\ref{fig:modules}), far from the behaviour in the unitary case. 
The elements of the example zero-reality basis given above, contrarily, satisfy $SH\ketbra{w}(SH)^{\dagger}\in{\rm span}_{\mathbb{R}} \{I,Y\}^{\otimes n}$, and evolve roughly evenly over both the even- and odd-$Y$ modules, behaving therefore ``closer'' to the unitary case. Quantitatively, this is reflected in the fact that the latter state distribution forms a unitary state 3-design -- exactly as required for identical shadow protocols.  

One can also consider the impact of measuring in a complex basis in the case of local orthogonal shadows. Here, interestingly, the reality need not be fixed between each local basis, with individual qubits capable of being treated differently. More generally, one can imagine mixing and matching local unitary and orthogonal shadows. For example, if the observables in which one is interested in learning is a set of Pauli strings that can take any value on certain qubits but only $I, X $ and $Z$ on others (say, $k'$ of them), one could choose to evolve with local unitaries on the former, and local orthogonals on the later. This yields a reduction in sample complexity over the local unitary protocol by a factor of $(3/2)^{k'}$, while still allowing for unbiased estimates of the target observables.

\section{Discussion}
In this work we have closed a notable gap in the classical shadows literature by characterizing the performance of orthogonal  shadows, the last protocol resulting from a classical compact Lie group that was yet to be explored.  Pleasingly, beyond filling this conceptual gap, we have found that local real shadows can in fact, for the class of observables containing only $I,\ X $ and $Z$ in their Pauli decompositions, drastically outperform previously known schemes. This result has crucial implications in for example  studying properties of condensed matter and strongly-interacting systems, as here one seeks to estimate the expectation value of Hamiltonians that are generally local and real-valued~\cite{cerezo2017factorization}, or when mapping fermionic operators to qubit Hamiltonians via means of Bravyi-Kitaev transformations, which can produce real-valued operators with $\mco(\log n)$-bodyness~\cite{seeley2012bravyi}.  
In addition to the direct benefit of now being able to more efficiently learn these operators, this highlights the promise of studying further distributions, to uncover yet more classes of efficiently learnable operators. 

We emphasize that the advantages we report here refer exclusively to \textit{sample}-efficiency; an important question that remains unanswered is the extent to which real shadows can reproduce the favourable resource scaling of unitary shadows in other respects~\cite{bertoni2024shallow,zhang2024minimal,park2023resource,schuster2024random}, in particular the recently-discovered fact that they can be (approximately) implemented in logarithmic depth~\cite{schuster2024random}. Counterintuitively, it is unclear that the key fact underpinning this -- the existence of logarithmic-depth approximate unitary designs -- will be reproduced in the case of the orthogonal group. For example, it is known that no such approximate design can exist if its elements are themselves orthogonal~\cite{schuster2024random}. 

More generally, our findings support an emerging lesson of the difficulty of learning operators which live in very large spaces~\cite{ ragone2023unified,fontana2023adjoint,diaz2023showcasing,cerezo2023does,goh2023lie,bermejo2024quantum}. Most directly, in the case of local random unitary and orthogonal gates, the variances are proportional to the dimension of the space in which the operators to be measured have support. In the case of global real shadows, our asymptotic improvement of a factor of two over the case of global unitaries again reflects the relative difference between the dimensions of the corresponding target spaces. Given a class of operators that one wishes to measure but cannot be efficiently handled by any known classical shadows protocol, then, a natural suggestion is to look for a distribution whose action carries the desired operators into a space with as small a dimension as possible. For example, in the presence of symmetry, one can explicitly reduce the dimension of the visible space by focusing exclusively on symmetry respecting operators~\cite{sauvage2024classical}. In the future we expect that such tailored schemes will continue to be developed, continuing to grow the set of properties of quantum systems that may be efficiently learned.\\

\section{Acknowledgments}
The authors thank Manuel G. Algaba and Diego García-Martín for useful discussions. MW acknowledges the support of the Australian government research training program scholarship and the IBM Quantum Hub at the University of Melbourne. AAM acknowledges support by the German Federal Ministry for Education and Research (BMBF) under the project FermiQP. MW and AAM were supported by the U.S. DOE through a quantum computing program sponsored by the Los Alamos National Laboratory (LANL) Information Science \& Technology Institute. ML and MC acknowledge support by the Laboratory Directed Research and Development (LDRD) program of LANL under project numbers 20230049DR and 20230527ECR. ML was also supported by  the Center for Nonlinear Studies at LANL. MC also acknowledges initial support by from LANL ASC Beyond Moore’s Law project.

\bibliography{refs,quantum}

\onecolumngrid
\appendix

\clearpage
\newpage

\section*{Appendices for ``\textit{R\MakeLowercase{eal classical shadows}}'' }

\section{Preliminaries}

In this section we briefly review the basics of the Weingarten calculus, as these will be used to derived our main results.  An accessible introduction to the relevant mathematics may be found in Ref.~\cite{mele2023introduction}.

The core technical step in proving most of our results is the evaluation of integrals  of the form
\begin{equation}\label{eq:momop}
\mct^{(k)}_G (A) =\int_{U\sim \mathcal{U}} U^{\otimes k} AU^{\dagger \otimes k} = \int_{U\in G}d\mu(U) U^{\otimes k} AU^{\dagger \otimes k}\,.
\end{equation}
Here, we have assumed that $\mathcal{U}$ is the Haar measure $d\mu$ over some group  $G$ and $A\in\mcl\left(\mathcal{H}^{\otimes k} \right)$ is some arbitrary linear operator on $\mathcal{H}^{\otimes k} $. For our purposes, we will take $G\in \{{\rm O}(d),{\rm O}(2)^{\otimes n}\}$ and be interested in the $k=2$ and $k=3$, cases. As is well-known~\cite{mele2023introduction}, the (super) operator $\mct^{(k)}_G$ orthogonally projects onto the $k$-th order commutant of $G$, i.e., onto the set
\begin{equation}
{\rm comm}\left(G, k\right) = \big\{ B\in \mcl\left(\mathcal{H}^{\otimes k} \right) \ \vert \ [B,U^{\otimes k}]=0\ \forall U\in G \big\} \subseteq \mathcal{L}(\mathcal{H}^{\otimes k} ).
\end{equation}
Hence, given a spanning set of ${\rm comm}\left(G, k\right) $ one can evaluate integrals of the form Eq.~\eqref{eq:momop} by means of elementary linear algebra. 

For the case of $G={\rm O}(d)$, it is well known that a basis of the $k$-th order commutant is given by a  representation $F_d$ of the \textit{Brauer algebra} $\mathfrak{B}_k$~\cite{hashagen2018real,garcia2023deep,larocca2022group}. We recall that $\mathfrak{B}_k$ enumerates the two-element pairings of a set of $2k$ objects; an element $\sigma\in\mathfrak{B}_k$ is therefore characterized by a set of $k$ disjoint pairs, $\sigma=\{\{\lm_1,\sigma(\lm_1)\},\dots,\{\lm_k,\sigma(\lm_k)\}\}$. In the graphical  notation~\cite{mele2023introduction}, we can interpret these as labelling the beginning and endpoints of $k$ wires (see Fig.~\ref{fig:comm}). Explicitly, this representation is given by
\begin{equation*}
F_d(\sigma) = \sum_{i_1,\ldots,i_{2k}=0}^{d-1}\ketbra{i_{k+1},i_{k+2},\ldots,i_{2k}}{i_1,i_{2},\ldots,i_{k}}\prod_{\gamma=1}^k\delta_{i_{\lm_\gamma},i_{\sigma(\lm_\gamma)}},
\end{equation*}
a subalgebra of the polynomial ring $\mathbb{Z}[d]$~\cite{garcia2023deep}.
From this expression we can read off the form of the commutants for $k=2,3$. Concretely we have
\begin{equation}\label{eq:oc2}
  {\rm comm}\left({\rm O}(d), 2\right) = {\rm span}\{\id,\mbs,\om\}\,, 
\end{equation}
 with $\mbs$ the SWAP operator and $\ket{\Omega}=\sum_{i=0}^{d-1}\ket{ii}$, and
\begin{align}\label{eq:oc3}
  {\rm comm}\left({\rm O}(d), 3\right) = {\rm span}\{ \mbs_\id, \mbs_{(12)}, \mbs_{(13)}, &\mbs_{(23)}, \mbs_{(123)}, \mbs_{(132)},\nonumber\\ \Omega_{12;12}, \Omega_{23;23}, &\Omega_{13;13}, \Omega_{12;23}, \Omega_{23;12}, \Omega_{13;23}, \Omega_{13;12}, \Omega_{23;13}, \Omega_{12;13}\} \,. 
\end{align}
Here $\om_{ab;xy}=\sum_{i,j}\ket{i}_a\ket{i}_b\bra{j}_x\bra{j}_y$ denotes an (unnormalised) maximally entangled state across two copies of $\mbc^d$ (see Fig.~\ref{fig:comm}), and $\mbs_\sigma$ the natural action of the permutation $\sigma$ on $(\mbc^d)^{\otimes k}$, i.e.
\begin{equation*}
\mbs_\sigma = \sum_{i_1,\ldots,i_{k}=0}^{d-1}\ketbra{i_{\sigma^{-1}(1)},\ldots,i_{\sigma^{-1}(k)}}{i_1,\ldots,i_{k}}.
\end{equation*}
Combined with the above assertion that $\mct^{(k)}_G$ is the orthogonal projection onto the commutant,
\begin{equation}\label{eq:proj}
  \mct^{(k)}_G : \mathcal{L}((\mbc^d)^{\otimes k})  \xrightarrow[]{}  {\rm comm}\left(G, k\right)\,,
\end{equation}
Eqs.~\eqref{eq:oc2} and~\eqref{eq:oc3} 
will allow us to characterize the global and local orthogonal shadow protocols.

\begin{figure}
    \centering
    \includegraphics[width=0.7\textwidth]{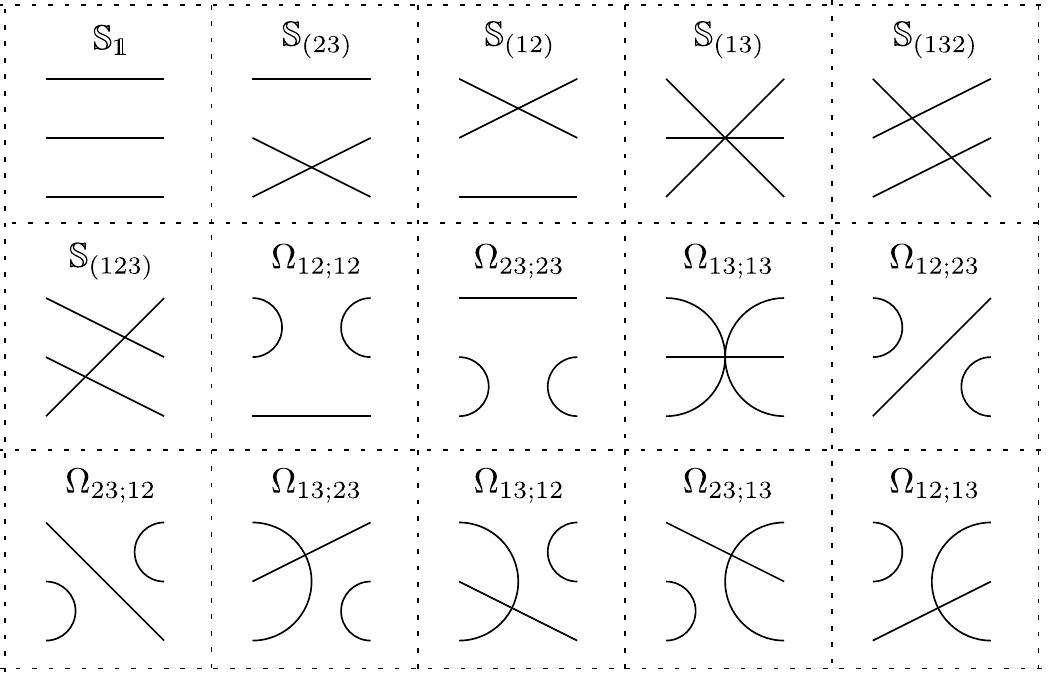}
    \caption{Graphical representation of the 15 elements of the third order commutant ${\rm comm}\left({\rm O}(d), 3\right)$ of the orthogonal group (see Ref.~\cite{mele2023introduction} for an introduction to this notation).} 
    \label{fig:comm}
\end{figure}

\section{Proof of main results}

From the discussion in the main text (in particular Eqs.~\eqref{eq:mc2},~\eqref{eq:ea2} and~\eqref{eq:e2a}) we know that we will need to evaluate $\mct^{(k)}_{{\rm O}(d)}(\Pi_w^{\otimes k})$ for $k=2,3$, which we do over the following two lemmas.

\begin{lemma}\label{lem:gw2}
Let  $\Pi_w=\ketbra{w}$ be any rank-1 projector. Then, we have
\begin{equation}
   \int_{U\in {\rm O}(d)}d\mu(U)\   U^{\dagger\otimes 2}\Pi_w^{\otimes 2}U^{\otimes 2} = \frac{(d-\alpha_w)\id + (d-\alpha_w)\mbs + (\alpha_w d+\alpha_w-2)\om}{d(d-1)(d+2)}   \label{eq:g2c}
  \end{equation}
where $\alpha_w=\abs{\braket{w|w^*}}^2$.
\end{lemma}

\begin{proof}
  By Eqs.~\eqref{eq:oc2} and~\eqref{eq:proj} we have
  \begin{equation}
    \int_{U\in {\rm O}(d)} d\mu(U)\  U^{\dagger\otimes 2}\Pi_w^{\otimes 2}U^{\otimes 2} = c_{\id}\id + c_{\mbs}\mbs + c_{\om}\om\,,
    \label{eq:span}
  \end{equation}
    for some coefficients $c_{\id},c_{\mbs},c_{\om}$  (we suppress the $d$-dependence of these coefficients to simplify the notation). Multiplying Eq.~\eqref{eq:span} on both sides by $\id,\mbs$ and $\om$, respectively, and then taking the trace, we arrive at the system of equations
    \begin{equation}
    \begin{aligned}
   1  &= c_{\id}d^2 + c_{\mbs}d + c_{\om}d \\
    1  &= c_{\id}d + c_{\mbs}d^2 + c_{\om}d \\
    \alpha_w &= c_{\id}d + c_{\mbs}d + c_{\om}d^2.
    \end{aligned}
    \end{equation}
    Here we have used $(\ketbra{\Omega})(\ketbra{\Omega})=d\ketbra{\Omega}$ and $ \Tr[\Pi_w^{\otimes 2}\ketbra{\Omega}]=\Tr[\Pi_w(\Pi_w)^\mathsf{T}]=\abs{\braket{w|w^*}}^2=\alpha_w$.
Solving for the coefficients yields the claimed result.
\end{proof}

This strategy of expressing $\mct^{(k)}_{{\rm O}(d)}(\Pi_w^{\otimes k})$ as a linear combination of elements in ${\rm comm}\left(G, k\right)$ and then deducing the coefficients by multiplying both sides with operators appearing in the sum before taking the trace will be reused in the next lemma.

\begin{lemma}\label{lem:gw3}
Let $\Pi_w=\ketbra{w}$ be any rank-1 projector. Then, we have
\begin{equation}\label{eq:gco3result}
  \int_{U\in {\rm O}(d)} d\mu(U)\   U^{\dagger\otimes 3}\Pi_w^{\otimes 3}U^{\otimes 3} = \frac{(d-3\alpha_w+2)\sum_{\mbs_{\pi}\in \mathcal{S}} \mbs_{\pi}+(\alpha_w d+\alpha_w-2)\sum_{\Omega_x\in \mathcal{X} } \Omega_x }{d(d-1)(d+2)(d+4)}  \,,
  \end{equation}
  where 
\[\mathcal{S}=\set{\mbs_\id, \mbs_{(12)}, \mbs_{(13)}, \mbs_{(23)}, \mbs_{(123)}, \mbs_{(132)}},\]
  \[ \mathcal{X}=\set{\Omega_{12;12}, \Omega_{23;23}, \Omega_{13;13}, \Omega_{12;23},\Omega_{23;12}, \Omega_{13;23}, \Omega_{13;12}, \Omega_{23;13}, \Omega_{12;13}}\]
  and $\alpha_w=\abs{\braket{w|w^*}}^2$.
\end{lemma}

The following proof  will be entirely analogous to that of Lemma~\ref{lem:gw2}, albeit resulting in a larger system of simultaneous equations.
\begin{proof}
  By Eqs.~\eqref{eq:oc3} and~\eqref{eq:proj} we have:
  \begin{equation}
  \int_{U\in {\rm O}(d)} d\mu(U)\  U^{\dagger\otimes 3}\Pi_w^{\otimes 3}U^{\otimes 3} = \sum_{\pi\in S_3}a_\pi \mbs_{\pi}+\sum_x b_x \Omega_x 
    \label{eq:c3span1}
  \end{equation}
  for some coefficients $a_\pi$ and $b_x$ to be determined (where we continue to suppress the $d$-dependence of the coefficients to simplify the notation). By the permutation symmetry of Eq.~\eqref{eq:c3span1} we immediately have
 \begin{align*}
   &a_{(12)} = a_{(13)} = a_{(23)};\qquad a_{(123)} = a_{(132)};\qquad \\ &
   b_{12;12} = b_{13;13} = b_{23;23};\hspace{4mm} b_{12;13} = b_{12;23} = b_{13;12} = b_{13;23} =b_{23;12} =b_{23;13} 
  \end{align*}
 and so 
 \small
  \begin{align}
    \int_{U\in {\rm O}(d)} d\mu(U)\ U^{\dagger\otimes 3}\Pi_w^{\otimes 3}U^{\otimes 3} =& a_{\id}\id + a_{(12)}\left(\mbs_{(12)} + \mbs_{(13)} + \mbs_{(23)}\right) + a_{(123)}\left( \mbs_{(123)}+ \mbs_{(132)}\right)\nonumber\\
    &+b_{12;12}\left( \Omega_{12;12}+ \Omega_{23;23}+ \Omega_{13;13}\right)+ b_{12;13}\left(\Omega_{12;23}+\Omega_{23;12}+ \Omega_{13;23}+ \Omega_{13;12}+ \Omega_{23;13}+ \Omega_{12;13}\right)\,.
    \label{eq:c3span}
  \end{align}
  \normalsize
  As in the proof of Lemma~\ref{lem:gw2} we will evaluate the coefficients by multiplying Eq.~\eqref{eq:c3span} by various operators appearing on the right hand side, and then taking the trace.
  We begin by noting
  \small
  \begin{align*}
   &\Tr\left[\id\right]=d^3;\hspace{5mm}\Tr\left[\mbs_{(12)}\right]=\Tr\left[\mbs_{(13)}\right]=\Tr\left[\mbs_{(23)}\right]=d^2;\hspace{5mm}\Tr\left[\mbs_{(123)}\right]=\Tr\left[\mbs_{(132)}\right]=d;\\
&\Tr\left[\Omega_{12;12}\right]=\Tr\left[\Omega_{13;13}\right]=\Tr\left[\Omega_{23;23}\right]=d^2;\hspace{5mm}\Tr\left[\Omega_{12;13}\right]=\Tr\left[\Omega_{12;23}\right]=\Tr\left[\Omega_{13;12}\right]=\Tr\left[\Omega_{13;23}\right]=\Tr\left[\Omega_{23;12}\right]=\Tr\left[\Omega_{23;13}\right]=d\,.
  \end{align*}
  \normalsize
Multiplying both sides of Eq.~\eqref{eq:c3span} by $\id,\mbs_{(12)},\mbs_{(123)},\Omega_{12;12}$ and $\Omega_{12;13}$, respectively, we arrive at the system of equations
\begin{equation}
\begin{aligned}
1&=d^{3}a_{\id}  +3d^{2}a_{(12)}  +2d^{}a_{(123)} + 3d^{2} b_{12;12}+ 6db_{12;13}\\
1&=d^{2}a_{\id}  +(d^{3}+2d)a_{(12)}  +2d^{2}a_{(123)} + (d^{2}+2d) b_{12;12}+ (2d^2+4d)b_{12;13}\\
1&=d^{}a_{\id}  +3d^{2}a_{(12)}  +(d^{3}+d)a_{(123)} + 3d b_{12;12}+ (3d^2+3d)b_{12;13}\\
\alpha_w &= d^{2} a_{\id} + (2d + d^{2}) a_{(12)} + 2d\, a_{(123)} + (d^{3} + 2d) b_{12;12} + (4d^{2} + 2d) b_{12;13}, \\
\alpha_w &= d\, a_{\id} + (2d + d^{2}) a_{(12)} + (d^{2} + d) a_{(123)} + (2d^{2} + d) b_{12;12} + (d^{3} + 2d^{2} + 3d) b_{12;13}.
\end{aligned}
\end{equation}
  Solving this system we find that the result is as claimed. 
\end{proof}

With Lemma~\ref{lem:gw2} in hand it is  straightforward to evaluate the global orthogonal Clifford measurement channel. We will first prove Proposition~\ref{prop:cmc}, of which Proposition~\ref{prop:gmc} is a special case.

\cmc*

\begin{proof}
Starting from the characterization Eq.~\eqref{eq:mc1} of the measurement channel we can simply directly calculate:
\begin{align}
  \mcm_{{\rm O}(d);\alpha}(A)&= \int_{U\in {\rm O}(d)} d\mu(U) \sum_w  \Tr\left[A U^\dagger\Pi_wU\right] U^\dagger\Pi_w U \nonumber \\
  &=  \Tr_1 \left[(A \otimes \id)\sum_w\int_{U\in {\rm O}(d)} d\mu(U)\  U^{\dagger\otimes 2}\Pi_w^{\otimes 2}U^{\otimes 2}\right]\nonumber \\
  &=  \Tr_1 \left[(A \otimes \id)\sum_w \left( \frac{(d-\alpha_w)\id + (d-\alpha_w)\mbs + (\alpha_w d+\alpha_w-2)\om}{d(d-1)(d+2)}  \right)\right] \label{eq:g2cu} \\
  &=  \Tr_1 \left[(A \otimes \id)\left( \frac{(d^2-\alpha)\id + (d^2-\alpha)\mbs + (\alpha d+\alpha-2d)\om}{d(d-1)(d+2)}  \right)\right] \nonumber \\
  &=  \frac{\Tr[A](d^2-\alpha)\id+(d^2-\alpha)A +(\alpha d+\alpha-2d)A^\mathsf{T} }{d(d-1)(d+2)}\nonumber\,,
\end{align}
where we have recalled the \textit{reality} $\alpha=\sum_w \alpha_w$ of the basis and
in Eq.~\eqref{eq:g2cu} used the result Eq.~\eqref{eq:g2c} of Lemma~\ref{lem:gw2}. Further recalling the definition 
\begin{equation*}
  \mcd_{p}(A) = \frac{p\Tr[A]}{d}\id + (1-p)A \,, 
\end{equation*}
of the depolarizing channel of strength $p$ we then have, with $p_\alpha=(d^2-\alpha)/((d-1)(d+2))$, 
\begin{align*}
\mcm_{{\rm O}(d);\alpha}(A)&= \frac{\Tr[A]p_\alpha\id}{d}+\frac{(1-p_\alpha)}{(1-p_\alpha)}\frac{(d^2-\alpha)A +(\alpha d+\alpha-2d)A^\mathsf{T} }{d(d-1)(d+2)}\\
&=\frac{\Tr[A]p_\alpha\id}{d}+(1-p_\alpha)\frac{(d-1)(d+2)}{d-2+\alpha}\frac{(d^2-\alpha)A +(\alpha d+\alpha-2d)A^\mathsf{T} }{d(d-1)(d+2)}\\
&=\frac{\Tr[A]p_\alpha\id}{d}+(1-p_\alpha)\frac{(d^2-\alpha)A +(\alpha d+\alpha-2d)A^\mathsf{T} }{d(d-2+\alpha)}\\
&= \mcd_{p_\alpha}(\widetilde{A}_\alpha)
\end{align*}
with 
\begin{equation}\label{eq:atwiddle}
\widetilde{A}_\alpha=\frac{(d^2-\alpha)A +(\alpha d+\alpha-2d)A^\mathsf{T} }{d(d-2+\alpha)}\,.
\end{equation}
\end{proof}
If the basis is real then we have $\alpha_w=1\ \forall w$, and so $\alpha=d$. Making this substitution in Proposition~\ref{prop:cmc} proves Proposition~\ref{prop:gmc}.\\ 

As discussed in the main text, the general case of a non-real basis interpolates between real and unitary classical shadows. The reality $\alpha$ of a given basis is in general extensive, $\alpha\sim\mco(d)$; defining $\alpha=fd$,  we have
\begin{equation*}
\widetilde{A}_{fd}=\frac{(d-f)A +(f d+f-2)A^\mathsf{T} }{d-2+fd} \xrightarrow[d\to\infty]{} \frac{A +f A^\mathsf{T} }{f+1}.
\end{equation*}
For all values of $d$, and to simplify the notation defining 
\begin{equation*}
    q_\alpha=\frac{d^2-\alpha  }{d(d-2+\alpha)},\quad q'_\alpha=1-q_\alpha,
\end{equation*}
we have (note $q_\alpha\geq q'_\alpha$)
\begin{align*}
\mcm_{{\rm O}(d);\alpha}(A)&= \frac{\Tr[A]p_\alpha\id}{d}+(1-p_\alpha)(q_\alpha A +q'_\alpha A^\mathsf{T})\\
&= (q_\alpha+q'_\alpha)\frac{\Tr[A]p_\alpha\id}{d}+(1-p_\alpha)((q_\alpha+q'_\alpha-q'_\alpha)A +q'_\alpha A^\mathsf{T})\\
&= q_\alpha\frac{\Tr[A]p_\alpha\id}{d} + q'_\alpha\left[ \frac{\Tr[A]p_\alpha\id}{d} + (1-p_\alpha)(A +A^\mathsf{T}) \right] +(1-p_\alpha)(q_\alpha-q'_\alpha)A \\
&= (q_\alpha-q'_\alpha)\frac{\Tr[A]p_\alpha\id}{d} + q'_\alpha\left[ \frac{2\Tr[A]p_\alpha\id}{d} + (1-p_\alpha)(A +A^\mathsf{T}) \right] +(1-p_\alpha)(q_\alpha-q'_\alpha)A \\
&= (q_\alpha-q'_\alpha)\frac{\Tr[A]p_\alpha\id}{d} + 2q'_\alpha\left[ \frac{\Tr[A]p_\alpha\id}{d} + (1-p_\alpha)\frac{(A +A^\mathsf{T})}{2} \right] +(1-p_\alpha)(q_\alpha-q'_\alpha)A \\
&= (q_\alpha-q'_\alpha)\mcd_{p_\alpha} (A) + 2q'_\alpha\mcd_{p_\alpha}(A_{\rm sym})\,.
\end{align*}
That is, the channel constitutes a tunable linear combination of real and unitary classical shadows, depending through $q_\alpha$ and $q'_\alpha$ on the reality $\alpha$ of the measurement basis. For example, let us pick a random basis by taking the columns of a Haar random unitary $U$. We then have
\begin{align*}
\underset{U\sim\mu_\mbu}{\mbe} \alpha &=\underset{U\sim\mu_\mbu}{\mbe}\sum_w \alpha_w=\underset{U\sim\mu_\mbu}{\mbe} \sum_w \abs{\braket{w|w^*}}^2 =\underset{U\sim\mu_\mbu}{\mbe} \sum_{i=1}^d \abs{\sum_{j=1}^d U_{i,j}^2}^2 =d\underset{U\sim\mu_\mbu}{\mbe} \abs{\sum_{j=1}^d U_{1,j}^2}^2 \\
&=d\underset{U\sim\mu_\mbu}{\mbe} \sum_{i,j=1}^d U_{1,i}^2 (U_{1,j}^*)^2 =d\underset{U\sim\mu_\mbu}{\mbe} \sum_{i=1}^d |U_{1,i}|^4 =d^2\underset{U\sim\mu_\mbu}{\mbe}|U_{1,1}|^4 =\frac{2d}{d+1}\,,
\end{align*}
where we have used $\mbe_{U\sim\mu_{\mbu(d)}}|U_{i,j}|^4=2/(d^2+d)\ \forall 1\leq i,j\leq d$.
So for a random basis, as $d\to\infty$ we have $f= 2/(d+1)\to 0,\ q\to 1,\ q'\to 0$, recovering the global unitary shadows channel.\\

\noindent
With the help of Lemma~\ref{lem:gw3} we can readily evaluate the variance induced by the global orthogonal channels. For simplicity we begin with the real case:

\gv*

\begin{proof}
We begin by noting that~\cite{huang2020predicting}
\begin{align*}
{\rm Var}[\hat{a}] &= \mbe\left[\left(\hat{a}-\mbe[\hat{a}]\right)^2\right]\\
&=\mbe\left[\left(\Tr[A\hat{\rho}]-\mbe[\Tr[A\hat{\rho}]]\right)^2\right]\\
&=\mbe\left[\left(\Tr[(A_0 + \Tr[A]\id/d)\hat{\rho}]-\mbe[\Tr[(A_0 + \Tr[A]\id/d)\hat{\rho}]]\right)^2\right]\\
&=\mbe\left[\left(\Tr[A_0\hat{\rho}]+ \Tr[A]\Tr[\hat{\rho}]-\mbe[\Tr[A_0\hat{\rho}]]-\Tr[A]\mbe[\Tr[\hat{\rho}]]\right)^2\right]\\
&=\mbe\left[\left(\Tr[A_0\hat{\rho}]-\mbe[\Tr[A_0\hat{\rho}]]\right)^2\right]\\
&=\mbe\left[\left(\hat{a}_0-\mbe[\hat{a}_0]\right)^2\right]\\
&={\rm Var}[\hat{a}_0] \,,
\end{align*}
where $A_0=A-\Tr[A]\id/d$ is the traceless component of $A$, and we have used that $\Tr[\hat{\rho}]=1$ for all classical shadows $\hat{\rho}$. This implies that we can consider only the traceless component of the observables we wish to measure without affecting the variance, which will reduces the number of non-zero terms in the following calculations. We now evaluate the two terms on the right hand side of
\begin{equation}\label{eq:vardef}
 {\rm Var}[\hat{a}_0] = \mbe\left[\hat{a}_0^2\right]-\mbe\left[\hat{a}_0\right]^2.
\end{equation}
The first term is:
  \begin{align}
    \mbe[\hat{a}_0^2] &= \sum_w\int_{U\in {\rm O}(d)} d\mu(U)\Tr\left[\rho U^\dagger\Pi_wU\right]\Tr\left[A_{ 0} \mcm^{-1}\left(U^\dagger\Pi_wU\right)\right] ^2 \nonumber \\
     &= \sum_w\int_{U\in {\rm O}(d)} d\mu(U)   \Tr\left[\rho U^\dagger\Pi_wU\right]\Tr\left[\mcm^{-1}\left(A_{ 0}\right) U^\dagger\Pi_wU\right] ^2 \nonumber \\
     &= \frac{(d+2)^2}{4}\sum_w\int_{U\in {\rm O}(d)} d\mu(U)   \Tr\left[\rho U^\dagger\Pi_wU\right]\Tr\left[A_{{\rm sym}; 0} U^\dagger\Pi_wU\right] ^2\label{eq:dpinv}  \\
     &= \frac{(d+2)^2}{4}\sum_w\int_{U\in {\rm O}(d)} d\mu(U)   \Tr\left[\left(\rho\otimes A_{{\rm sym}; 0}\otimes A_{{\rm sym}; 0}\right)\left(U^{\dagger\otimes 3}\Pi_w^{\otimes 3}U^{\otimes 3}\right) \right] \nonumber\\
     &= \frac{(d+2)^2}{4}\Tr\left[\left(\rho\otimes A_{{\rm sym}; 0}\otimes A_{{\rm sym}; 0}\right)\sum_w\int_{U\in {\rm O}(d)} d\mu(U) \left(U^{\dagger\otimes 3}\Pi_w^{\otimes 3}U^{\otimes 3}\right) \right] \nonumber\\
     &= \frac{(d+2)^2}{4}\Tr\left[\left(\rho\otimes A_{{\rm sym}; 0}\otimes A_{{\rm sym}; 0}\right)\sum_w\frac{\sum_{\pi\in S_3} \mbs_{\pi}+\sum_x \Omega_x }{d(d+2)(d+4)}   \right]\label{eq:usegc3} \\
     &= \frac{d+2}{4(d+4)}\Tr\left[\left(\rho\otimes A_{{\rm sym}; 0}\otimes A_{{\rm sym}; 0}\right)\left({\sum_{\pi\in S_3} \mbs_{\pi}+\sum_x \Omega_x }\right)   \right] \label{eq:15sumo}\\
     &= \frac{d+2}{4(d+4)}\left(2\Tr\left[ A_{{\rm sym}; 0}^2\right] + 8\Tr\left[\rho A_{{\rm sym}; 0}^2\right]\right) \nonumber  \\
     &= \frac{d+2}{2d+8}\left(\Tr\left[ A_{{\rm sym}; 0}^2\right] + 4\Tr\left[\rho A_{{\rm sym}; 0}^2\right]\right) \nonumber\,.
  \end{align}
In Eq.~\eqref{eq:dpinv} we have used the characterization of $\mcm$ as a depolarizing channel (Prop.~\ref{prop:gmc}) to easily  invert its action on (the traceless) $A_0$, and in Eq.~\eqref{eq:usegc3} the result of Lemma~\ref{lem:gw3}. Finally, one can readily verify that the 15 summands arising from Eq.~\eqref{eq:15sumo} consist of two copies of $\Tr\left[ A_{{\rm sym}; 0}^2\right]$, eight copies of $\Tr\left[\rho A_{{\rm sym}; 0}^2\right]$ and five copies of zero (due to factors of $\Tr[A_0]=0$). For the second term on the right hand side of Eq.~\eqref{eq:vardef} we have
  \begin{align*}
    \mbe[\hat{a}_0]^2 &=\left( \sum_w\int_{U\in {\rm O}(d)} d\mu(U)   \Tr\left[\rho U\dagger\Pi_wU\right]\Tr\left[A_{ 0} \mcm^{-1}\left(U^\dagger\Pi_wU\right)\right]\right) ^2  \\
     &= \left(\sum_w\int_{U\in {\rm O}(d)} d\mu(U) \Tr\left[\rho U^\dagger\Pi_wU\right]\Tr\left[\mcm^{-1}\left(A_{ 0}\right) U^\dagger\Pi_wU\right] \right)^2  \\
     &= \left(\Tr\left[\mcm^{-1}\left(A_{ 0}\right)\sum_w\int_{U\in {\rm O}(d)} d\mu(U) \Tr\left[\rho U^\dagger\Pi_wU\right] U^\dagger\Pi_wU\right] \right)^2  \\
     &= \Tr\left[\mcm^{-1}\left(A_{ 0}\right)\mcm(\rho)\right] ^2  \\
     &= \Tr\left[\mcm\left(\mcm^{-1}\left(A_{ 0}\right)\right)\rho\right] ^2  \\
     &= \Tr\left[A_{{\rm sym}; 0}\rho\right] ^2  \,,
  \end{align*}
  where we have used the self-adjointness of $\mcm$ with respect to the Hilbert-Schmidt inner product~\cite{huang2020predicting} and the fact that 
  \begin{equation*}
      \mcm\circ\mcm^{-1}=\id_{{\rm range}(\mcm)}={\rm proj}_{\{A\hspace{0.5mm}|\hspace{0.5mm}A=A^\mathsf{T}\}}\,.
  \end{equation*}
Altogether we therefore have
  \begin{equation*}
  {\rm Var}(\hat{a}) =\frac{d+2}{2d+8} \left(\Tr[A_{{\rm sym}; 0}^2] + 4\Tr[\rho A_{{\rm sym}; 0}^2]\right) - \Tr[A_{{\rm sym}; 0} \rho]^2 .
  \end{equation*}
\end{proof}

One can repeat essentially the same calculation as in the proof of Prop.~\ref{prop:gvar} in the more general case of a non-real basis of given reality $\alpha$; the result is, with $p_\alpha=(d^2-\alpha)/((d-1)(d+2))$ and $\widetilde{A}_\alpha$ as in Eq.~\eqref{eq:atwiddle},
\begin{align*}
{\rm Var}_{\alpha}(\hat{a}) =& \frac{1}{(1-p_\alpha)^2}\left( \frac{(d^2-3\alpha+2d)\left(\Tr[\widetilde{A}_{\alpha;0}^2] + 2\Tr[\rho\widetilde{A}_{\alpha;0}^2]  \right) }{d(d-1)(d+2)(d+4)}   \right)\\
&+\frac{1}{(1-p_\alpha)^2}\left( \frac{(\alpha d+\alpha-2d)\left( \Tr[\widetilde{A}_{\alpha;0}\widetilde{A}_{\alpha;0}^\mathsf{T}]+2\Tr[\rho\widetilde{A}_{\alpha;0}\widetilde{A}_{\alpha;0}^\mathsf{T}]+2\Tr[\rho\widetilde{A}_{\alpha;0}^\mathsf{T}\widetilde{A}_{\alpha;0}]+2\Tr[\rho\widetilde{A}_{\alpha;0}^\mathsf{T}\widetilde{A}_{\alpha;0}^\mathsf{T}]\right) }{d(d-1)(d+2)(d+4)}   \right) \nonumber\\
&- \Tr\left[\widetilde{A}_{\alpha;0}\rho\right] ^2  \,,
\end{align*}
where $\widetilde{A}_{\alpha;0}=\widetilde{A}_\alpha-\Tr[A]\id/d$. As discussed in the main text, for large $d$ we can typically (by Holder's inequality) neglect the terms which involve $\rho$. By the Cauchy-Schwarz inequality we then have
\begin{align*}
{\rm Var}_{\alpha}(\hat{a}) &\sim \frac{1}{(1-p_\alpha)^2}\left( \frac{(d^2-3\alpha+2d)\Tr[\widetilde{A}_{\alpha;0}^2]   + (\alpha d+\alpha-2d) \Tr[\widetilde{A}_{\alpha;0}\widetilde{A}_{\alpha;0}^\mathsf{T}] }{d(d-1)(d+2)(d+4)}   \right)\\
&\leq \frac{\|\widetilde{A}_0\|_2^2}{(1-p_\alpha)^2}\left( \frac{d^2 + \alpha d-2\alpha }{d(d-1)(d+2)(d+4)}   \right)\\
&= \|\widetilde{A}_0\|_2^2\frac{(d-1)^2(d+2)^2}{(d+\alpha-2)^2}\left( \frac{d^2 + \alpha d-2\alpha }{d(d-1)(d+2)(d+4)}   \right)\\
&\sim \frac{\|\widetilde{A}_0\|_2^2}{1+f} \,,
\end{align*}
where in the last step we have dropped all but the leading terms in $d$, and recalled the definition $f=\alpha/d$.

\lmc*

\begin{proof}
  Let 
  \begin{equation}
  A=    \sum_{i_1,\ldots,i_n\in\{\id,X,Y,Z\}}a_{i_1,\ldots,i_n}A_{i_1}\otimes\ldots\otimes A_{i_n}
  \end{equation}
  As in the case of local unitary Cliffords~\cite{huang2020predicting}, the measurement channel here factorises into a tensor product (following the factorisation of both the sampled operator $O\in {\rm O}(2)^{\otimes n}$ and the computational basis measurements). We therefore have
\begin{align*}
  \mcm(A)&=  \bigotimes_{j=1}^n \sum_{b_j\in\{0,1\}}\int_{U_j\in O(2)}d\mu(U_j)\sum_{i_1,\ldots,i_n\in\{\id,X,Y,Z\}}a_{i_1,\ldots,i_n}  \Tr\left[A_{i_j} U_j^\dagger \ketbra{b_j} U_j\right] U_j^\dagger\ketbra{b_j} U_j\\
  &=  \bigotimes_{j=1}^n \sum_{i_1,\ldots,i_n\in\{\id,X,Y,Z\}}a_{i_1,\ldots,i_n} \mcd_{1/2} \left((A_{i_j})_{\rm sym}\right)\\
  &=  \bigotimes_{j=1}^n \sum_{i_1,\ldots,i_n\in\{\id,X,Z\}}a_{i_1,\ldots,i_n} \mcd_{1/2} \left(A_{i_j}\right)\\
  &=  \mcd_{1/2}^{\otimes n} \big(A_{\rm L.S.}\big)\,.
\end{align*}
Where we have applied the result of Proposition~\ref{prop:gmc} in the special case $d=2$. The channel pseudo-inversion follows immediately. For instance, for locally symmetric $A=\bigotimes_j A_j$ we have
\begin{equation}
 \mcm^{-1}\bigg(\bigotimes_j A_j\bigg)=\bigotimes_{j=1}^n \left(2A_j - \frac{\Tr[A_j]}{2}\id\right),
  \label{eq:localoinv}
\end{equation}
with the classical shadows therefore being given by
\begin{equation}
  \hat{\rho} =\mcm^{-1}\left(U^\dagger |\hat{b}\rangle\langle\hat{b}| U\right)=\bigotimes_{j=1}^n \mcd_{1/2}^{-1}\left(U_j^\dagger |\hat{b}_j\rangle\langle\hat{b}_j| U_j\right) =\bigotimes_{j=1}^n \left(2U_j^\dagger |\hat{b}_j\rangle\langle\hat{b}_j| U_j - \frac{1}{2}\id\right)\,.
\end{equation}
\end{proof}

\lp*

\begin{proof}

For notational simplicity let $P=P_ {i_1}\otimes\ldots\otimes P_{i_k}\otimes\mbi^{n-k},\ P_{i_j}\in\{X,Z\}\ \forall 1\leq j\leq k$. As we will see next, the calculation is exactly the same irrespective of the $k$ over which the Pauli operators act non-trivially. From Eq.~\eqref{eq:localoinv} we have
\begin{align*}
\mcm^{-1}(P)&=\mcm^{-1}\left(P_{i_1}\otimes\ldots\otimes P_{i_k}\otimes\mbi^{n-k}\right)\\
&=\bigotimes_{j=1}^k \left(2P_{i_j} - \frac{\Tr[P_{i_j}]}{2}\id\right)\otimes \bigotimes_{j=k+1}^n \left(2\mbi - \frac{\Tr[\mbi]}{2}\id\right)\\
&=2^k \bigotimes_{j=1}^k \left(P_{i_j}\right)\otimes \bigotimes_{j=k+1}^n \mbi\\
&=2^kP\,.
\end{align*}
We can now directly calculate:  
\begin{align}
    {\rm Var}[\hat{p}] &\leq \max_{\rho} \sum_w\int_{U\in O(2)^{\otimes n}} d\mu(U)  \Tr\left[\rho U^\dagger\Pi_wU\right]\Tr\left[\bigg(\bigotimes_{j=1}^k P_j \otimes \mbi^{n-k}\bigg)\mcm^{-1}\left(U^\dagger\Pi_wU\right)\right] ^2\nonumber  \\
     &= \max_{\rho} \sum_w\int_{U\in O(2)^{\otimes n}} d\mu(U)   \Tr\left[\rho U^\dagger\Pi_wU\right]\Tr\left[\mcm^{-1}\bigg(\bigotimes_{j=1}^k P_j \otimes \mbi^{n-k}\bigg) U^\dagger\Pi_wU\right] ^2  \nonumber\\
     &= \max_{\rho} 4^k \Tr\Big[\rho\Big(\bigotimes_{j=1}^k \sum_{b_j\in\{0,1\}}\int_{U\in O(2)}d\mu(U)\ U^\dagger\Pi_{b_j}U\Tr\left[ P_j U^\dagger\Pi_{b_j}U\right] ^2 \nonumber\\
     &\quad\quad \quad\quad \quad\quad \quad\quad \times\bigotimes_{j'=k+1}^n \sum_{b_{j'}\in\{0,1\}}\int_{U\in O(2)} d\mu(U)\ U^\dagger\Pi_{b_{j'}}U\Tr\left[  U^\dagger\Pi_{b_{j'}}U\right] ^2 \Big)\Big] \nonumber\\
     &= \max_{\rho} 4^k \Tr\left[\rho\left(\bigotimes_{j=1}^k \sum_{b_j\in\{0,1\}}\int_{U\in O(2)}d\mu(U)\ U^\dagger\Pi_{b_j}U\Tr\left[ P_j U^\dagger\Pi_{b_j}U\right] ^2 \bigotimes_{j'=k+1}^n \int_{U\in O(2)}d\mu(U) U^\dagger\sum_{b_{j'}\in\{0,1\}}\Pi_{b_{j'}}U\right)\right]\nonumber \\
     &= \max_{\rho} 4^k \Tr\left[\rho\left(\bigotimes_{j=1}^k \sum_{b_j\in\{0,1\}}\int_{U\in O(2)}d\mu(U)\ U^\dagger\Pi_{b_j}U\Tr\left[ P_j U^\dagger\Pi_{b_j}U\right] ^2 \bigotimes_{j'=k+1}^n \mbi\right)\right]\nonumber \\
     &=  \max_{\rho} 4^k \Tr\left[(\rho\otimes\id\otimes\id)\bigotimes_{j=1}^k \sum_{b_j\in\{0,1\}}\frac{1}{48} \left(\mbi\otimes P_{j}\otimes P_j\right)\left({\sum_{\pi\in S_3} \mbs_{\pi}+\sum_x \Omega_x }\right)\right]\label{eq:48} \\
     &= \max_{\rho} 4^k \Tr\left[\rho\bigotimes_{j=1}^k \frac{1}{24} \left(2\Tr[P_j^2]\mbi + 8P_j^2\right)\right]\nonumber \\
     &= \max_{\rho} 4^k \Tr\left[ \frac{1}{2^k} \rho\right]\nonumber \\
     &= 2^k\nonumber\,.
  \end{align}
Here we have again used the self-adjointness of $\mcm$ with respect to the Hilbert-Schmidt inner product~\cite{huang2020predicting}, along with $\sum_{b_{j'}\in\{0,1\}}\Pi_{b_{j'}}=\mbi$. In Eq.~\eqref{eq:48} we have used our expression Eq.~\eqref{eq:gco3result} for the third order twirl of a rank one projector over O (whence the factor of $2(2+2)(2+4)=48)$. The sum of the 15 terms of Eq.~\eqref{eq:48} is evaluated as in the proof of Prop.~\ref{prop:gvar} (note $\Tr P_j = 0\ \forall 1\leq j\leq k$). 
\end{proof}

Before coming to the case of a more general local observable we prove the following lemma (c.f. Lemma 4 of Ref.~\cite{huang2020predicting}):

\begin{lemma}\label{lem:lr}
Given two $k$-qubit ``locally real'' Pauli observables $P_{\boldsymbol{p}}=P_{p_1}\otimes \cdots \otimes P_{p_k},\ P_{\boldsymbol{q}}=P_{q_1}\otimes \cdots \otimes P_{q_k}$, with $\boldsymbol{p},\boldsymbol{q}\in\{\mbi,X,Z\}^k$, we have, for any state $\rho$,
\begin{equation}
    \int_{U\in {\rm O}(2)^{\otimes k}}d\mu(U) \sum_w \Tr[\rho U^\dagger \Pi_w U]\Tr[(\mcd^{-1}_{1/2})^{\otimes k}(P_{\boldsymbol{p}}) U^\dagger \Pi_w U]\Tr[(\mcd^{-1}_{1/2})^{\otimes k}(P_{\boldsymbol{q}}) U^\dagger \Pi_w U]=f(\boldsymbol{p},\boldsymbol{q})\Tr[\rho P_{\boldsymbol{p}}P_{\boldsymbol{q}}]\,,
\end{equation}
where $f(\boldsymbol{p},\boldsymbol{q})=0$ whenever there exists an index $i$ such that $p_i  \neq q_i$ and $p_i, q_i  \neq\mbi$. Otherwise, $f(\boldsymbol{p},\boldsymbol{q})=2^s$, where $s$ is the number of indices where the corresponding Paulis both match and are not the identity, i.e. $s = \lvert\{i : p_i = q_i,\ p_i\neq \mbi\}\rvert$.
\end{lemma}

\begin{proof}
The proof is highly similar to that of the analogous result for local unitaries (as given in Ref.~\cite{huang2020predicting}) which we follow while making the appropriate adjustments. Now, by linearity and the respect that the inverse measurement channel $(\mcd^{-1}_{1/2})^{\otimes k}$ shows for tensor product  observables, the claim will readily follow if we can establish it in the single-qubit case. As such, we begin by considering some important subcases. Beginning with the case $P_p=P_q=\mbi$ we have:
\begin{align*}
\int_{U\in {\rm O}(2)}d\mu(U) \sum_w \Tr[\rho U^\dagger \Pi_w U]\Tr[\mcd^{-1}_{1/2}(\mbi) U^\dagger \Pi_w U]^2 &=\int_{U\in {\rm O}(2)} d\mu(U)\sum_w \Tr[\rho U^\dagger \Pi_w U]\Tr[ U^\dagger \Pi_w U]^2\\
&=\int_{U\in {\rm O}(2)} d\mu(U)  \Tr[\rho U^\dagger\sum_w \Pi_w U]\\
&=f(\mbi,\mbi)\Tr[\rho \mbi^2]\,,
\end{align*}
where we have used $\mcd^{-1}_{1/2}(\mbi)=\sum_w \Pi_w=\mbi$. Next, suppose $P_q=\mbi\neq P_p$ (the case $P_q\neq \mbi= P_p$ is similar).  Recalling Eq.~\eqref{eq:g2c} and that $\mcd^{-1}_{1/2}(P_{{p}})=2P_{{p}}$ we then see:
\begin{align*}
&\int_{U\in {\rm O}(2)}d\mu(U) \sum_w \Tr[\rho U^\dagger \Pi_w U]\Tr[\mcd^{-1}_{1/2}(P_{{p}}) U^\dagger \Pi_w U]\Tr[\mcd^{-1}_{1/2}(\mbi) U^\dagger \Pi_w U]\nonumber\\
&= \int_{U\in {\rm O}(2)}d\mu(U) \sum_w \Tr[\rho U^\dagger \Pi_w U]\Tr[2P_{{p}} U^\dagger \Pi_w U]\\
&= 2 \Tr[(\rho\otimes P_p) \sum_w \int_{U\in {\rm O}(2)} d\mu(U)\ U^{\dagger\otimes 2} \Pi_w^{\otimes 2} U^{\otimes 2}]\\
&= 2 \Tr[(\rho\otimes P_p) \left(\frac{\id + \mbs + \om}{2+2}  \right)]\\
&= f(P_p,\mbi) \Tr[\rho P_p \mbi]\,.
\end{align*}
Finally, we consider the case $P_p, P_q\neq \mbi$. This time we find (recalling Eq.~\eqref{eq:gco3result}):
\begin{align*}
&\int_{U\in {\rm O}(2)}d\mu(U) \sum_w \Tr[\rho U^\dagger \Pi_w U]\Tr[\mcd^{-1}_{1/2}(P_{{p}}) U^\dagger \Pi_w U]\Tr[\mcd^{-1}_{1/2}(P_{{q}}) U^\dagger \Pi_w U]\\
&=4\int_{U\in {\rm O}(2)} d\mu(U)\sum_w \Tr[\rho U^\dagger \Pi_w U]\Tr[P_{{p}} U^\dagger \Pi_w U]\Tr[P_{{q}} U^\dagger \Pi_w U]\\
&=4 \Tr[(\rho \otimes P_{{p}} \otimes P_{{q}}) \sum_w \int_{U\in {\rm O}(2)}d\mu(U)\ U^{\dagger\otimes 3} \Pi_w^{\otimes 3} U^{\otimes 3}]\\
&=4 \Tr[(\rho \otimes P_{{p}} \otimes P_{{q}})2\left( \frac{\sum_{\pi\in S_3} \mbs_{\pi}+\sum_x \Omega_x }{2(2+2)(2+4)} \right) ]\\
&=\frac{1}{6} \Tr[(\rho \otimes P_{{p}} \otimes P_{{q}})\left( \sum_{\pi\in S_3} \mbs_{\pi}+\sum_x \Omega_x \right) ]\\
&=\frac{1}{6} \Tr[\rho ( 2 \Tr[P_{{p}}  P_{{q}}]\mbi + 4P_{{p}}  P_{{q}}+4P_{{q}}  P_{{p}}) ]\\
&=2\delta_{p,q} \\
&=f({p},{q})\Tr[\rho P_{{p}}P_{{q}}]\,,
\end{align*}
where the various traces were evaluated by consulting the graphical representation of the commutant elements (see Fig.~\ref{fig:comm}), and recalling $X=X^\mathsf{T},\ Z=Z^\mathsf{T}$. We have also used $\{P_p,P_q\}=\Tr[P_{{p}}  P_{{q}}]=2\delta_{p,q} $. So the claim follows for all single-qubit Pauli strings, and therefore, as can readily be seen, for all Pauli strings. 
\end{proof}

\lo*
\begin{proof}
Here we will also closely follow the proof of the analogous result in Ref.~\cite{huang2020predicting}, making slight adjustments as required.  Let us consider an operator $A=\sum_{\boldsymbol{p}}a_{\boldsymbol{p}}P_{\boldsymbol{p}}$, where $P_{\boldsymbol{p}}\in\{\mbi,X,Z\}^k$ (as in previous results, the more general case in which the $P_{\boldsymbol{p}}$ act non-trivially on $k$ of $n$ total qubits is similar, and we focus on this special case for notational simplicity). We have:
\begin{align*}
{\rm Var}[\hat{a}] &\leq \max_{\rho} \int_{U\in {\rm O}(2)}d\mu(U) \sum_w \Tr[\rho U^\dagger \Pi_w U]\Tr[(\mcd^{-1}_{1/2})^{\otimes k}(A) U^\dagger \Pi_w U]^2\\
&=\max_{\rho} \int_{U\in {\rm O}(2)}d\mu(U) \sum_w \Tr[\rho U^\dagger \Pi_w U]\sum_{\boldsymbol{p},\boldsymbol{q}}a_{\boldsymbol{p}}a_{\boldsymbol{q}}\Tr[(\mcd^{-1}_{1/2})^{\otimes k}(P_{\boldsymbol{p}}) U^\dagger \Pi_w U]\Tr[(\mcd^{-1}_{1/2})^{\otimes k}(P_{\boldsymbol{q}}) U^\dagger \Pi_w U]\\
&=\max_{\rho}\Tr[\rho\sum_{\boldsymbol{p},\boldsymbol{q}}a_{\boldsymbol{p}}a_{\boldsymbol{q}}f({\boldsymbol{p}},{\boldsymbol{q}}) P_{\boldsymbol{p}}P_{\boldsymbol{q}}]\\
&=\left\|\sum_{\boldsymbol{p},\boldsymbol{q}}a_{\boldsymbol{p}}a_{\boldsymbol{q}}f({\boldsymbol{p}},{\boldsymbol{q}}) P_{\boldsymbol{p}}P_{\boldsymbol{q}}\right\|_{\infty}\,,
\end{align*}
where we have recalled the result of Lemma~\ref{lem:lr}, and that for any operator $O$ we have $\|O\|_\infty = \max_\rho \Tr[\rho O]$, where the  maximisation is over all states. To make further progress, as  in Ref.~\cite{huang2020predicting}, we introduce a partial order on Pauli strings $\boldsymbol{q},\boldsymbol{s}\in\{\mbi,X,Z\}^k$, writing $\boldsymbol{q}\triangleright\boldsymbol{s}$ if one can obtain $\boldsymbol{q}$ from $\boldsymbol{s}$ by replacing some local non-identity Paulis with $\mbi$. 
We note that there are $2^{k-|\boldsymbol{q}|}$ such strings in $\{X,Z\}^k$, where $|\boldsymbol{q}|$ is the number of qubits on which $\boldsymbol{q}$ acts non-trivially. 
Similarly, one can readily verify that the number of strings that dominate both $\boldsymbol{p}$ and $\boldsymbol{q}$ is $d({\boldsymbol{p}},{\boldsymbol{q}})=2^{k-|\boldsymbol{p}|-|\boldsymbol{q}|}f({\boldsymbol{p}},{\boldsymbol{q}})$. 
We may therefore continue:
\begin{align*}
\left\|\sum_{\boldsymbol{p},\boldsymbol{q}}a_{\boldsymbol{p}}a_{\boldsymbol{q}}f({\boldsymbol{p}},{\boldsymbol{q}}) P_{\boldsymbol{p}}P_{\boldsymbol{q}}\right\|_{\infty}&=\left\|\sum_{\boldsymbol{p},\boldsymbol{q}}a_{\boldsymbol{p}}a_{\boldsymbol{q}}d({\boldsymbol{p}},{\boldsymbol{q}})2^{-k+|\boldsymbol{p}|+|\boldsymbol{q}|} P_{\boldsymbol{p}}P_{\boldsymbol{q}}\right\|_{\infty}\\
&=\left\|2^{-k}\sum_{\boldsymbol{p},\boldsymbol{q}}a_{\boldsymbol{p}}a_{\boldsymbol{q}} \sum_{\boldsymbol{s}\hspace{0.5mm} \triangleleft \hspace{0.5mm}\boldsymbol{p},\boldsymbol{q}} 2^{|\boldsymbol{p}|+|\boldsymbol{q}|} P_{\boldsymbol{p}}P_{\boldsymbol{q}}\right\|_{\infty}\\
&=\left\|2^{-k}\sum_{\boldsymbol{s}\in\{X,Z\}^k }\left(\sum_{\boldsymbol{q} \hspace{0.5mm} \triangleright \boldsymbol{s}}   2^{|\boldsymbol{p}|} a_{\boldsymbol{p}}P_{\boldsymbol{p}}\right)^2\right\|_{\infty}\\
&\leq 2^{-k}\sum_{\boldsymbol{s}\in\{X,Z\}^k }\left(\sum_{\boldsymbol{q} \hspace{0.5mm} \triangleright \boldsymbol{s}}   2^{|\boldsymbol{q}|} |a_{\boldsymbol{q}}|\right)^2\,,
\end{align*}
where in the final line we have used the triangle inequality, the sub-multiplicativity of the spectral norm (i.e. $\|AB\|_\infty\leq\|A\|_\infty\|B\|_\infty$) and the fact that $\left\|P_{\boldsymbol{p}}\right\|_{\infty} = 1\ \forall \boldsymbol{p}$. Next, we notice that, for any $\boldsymbol{s}\in\{X,Z\}^k $,
\begin{equation*}
\sum_{\boldsymbol{q} \hspace{0.5mm} \triangleright \boldsymbol{s}} 2^{|\boldsymbol{q}|} = \sum_{j=0}^k {k \choose j}2^j = 3^k.
\end{equation*}
Applying the Cauchy-Schwarz inequality to the vectors $(2^{|\boldsymbol{q}|/2})_{\boldsymbol{q} \hspace{0.5mm} \triangleright \boldsymbol{s}}$ and $(2^{|\boldsymbol{q}|/2}|a_{\boldsymbol{q}}|)_{\boldsymbol{q} \hspace{0.5mm} \triangleright \boldsymbol{s}}$ then gives
\begin{align*}
2^{-k}\sum_{\boldsymbol{s}\in\{X,Z\}^k }\left(\sum_{\boldsymbol{q} \hspace{0.5mm} \triangleright \boldsymbol{s}}   2^{|\boldsymbol{q}|} |a_{\boldsymbol{q}}|\right)^2&\leq 2^{-k}\sum_{\boldsymbol{s}\in\{X,Z\}^k }\left(\sum_{\boldsymbol{q} \hspace{0.5mm} \triangleright \boldsymbol{s}}   2^{|\boldsymbol{q}|}\right)\left(\sum_{\boldsymbol{q} \hspace{0.5mm} \triangleright \boldsymbol{s}}   2^{|\boldsymbol{q}|} |a_{\boldsymbol{q}}|^2\right)\\
&=3^{k}\sum_{\boldsymbol{s}\in\{X,Z\}^k } \left(\sum_{\boldsymbol{q} \hspace{0.5mm} \triangleright \boldsymbol{s}}   2^{|\boldsymbol{q}|-k} |a_{\boldsymbol{q}}|^2\right)\\
&=3^{k}\sum_{\boldsymbol{q}}|a_{\boldsymbol{q}}|^2\\
&=3^{k}2^{-k}\|A\|_2^2\\
&\leq3^{k}\|A\|_\infty^2\,.
\end{align*}
Here we have used that 
\begin{equation*}
\|A\|_2^2=\left\|\sum_{\boldsymbol{p}}a_{\boldsymbol{p}}P_{\boldsymbol{p}}\right\|_2^2=\sum_{\boldsymbol{p,q}}a_{\boldsymbol{p}}a_{\boldsymbol{q}}^*\Tr[P_{\boldsymbol{p}}P_{\boldsymbol{q}}]=\sum_{\boldsymbol{p,q}}a_{\boldsymbol{p}}a_{\boldsymbol{q}}^*(2^k\delta_{\boldsymbol{p},\boldsymbol{q}})=2^k\sum_{\boldsymbol{p}}|a_{\boldsymbol{p}}|^2 \,.
\end{equation*}
Putting it all together we therefore conclude
\begin{equation*}
    {\rm Var}[\hat{a}] \leq 3^k \|A\|_{\infty}^2\,.
\end{equation*}
\end{proof}

\section{On the relation of the global real and unitary shadows variances}\label{sec:appc}
Let us consider the estimation of some symmetric self-adjoint operator $A$, with traceless component $A_0$. The exact variance expressions in the global unitary~\cite{huang2020predicting} and real (Prop.~\ref{prop:gvar}) cases are respectively
\begin{equation}
    {\rm Var}_{\rm Unitary} = \frac{d+1}{d+2}\left(\tr[A_0^2] + 2\tr[\rho A_0^2]\right) - \tr[\rho A_0]^2
\end{equation}
and
\begin{equation}
    {\rm Var}_{\rm Real} = \frac{d+2}{2d+8}\left(\tr[A_0^2] + 4\tr[\rho A_0^2]\right) - \tr[\rho A_0]^2.
\end{equation}
For $d\gg 1$ (but making no assumptions about the relevant sizes of $\tr[A_0^2]$ and $\tr[\rho A_0^2]$, i.e. not assuming $\|A_0\|_2\gg \|A_0\|_\infty$) we see that
\begin{align}
    {\rm Var}_{\rm Unitary} &\approx  \tr[A_0^2] + 2\tr[\rho A_0^2]-\tr[\rho A_0]^2,\\
    {\rm Var}_{\rm Real} &\approx \frac{1}{2} \tr[A_0^2] +  2\tr[\rho A_0^2]-\tr[\rho A_0]^2.
\end{align}
So certainly\footnote{Note ${\rm Var}_{\rm Unitary}-{\rm Var}_{\rm Real}=(1/2)\tr[A_0^2]\geq 0$} ${\rm Var}_{\rm Real} \leq {\rm Var}_{\rm Unitary}$. Using $\tr[A_0^2]=\|A_0\|_2^2  $ and $\tr[\rho A_0]\leq \|A_0\|_\infty$, in the main text we dropped the later two terms in the above equations, obtaining the claimed factor-of-two increase in sample-complexity,  justified when $\|A_0\|_2\gg \|A_0\|_\infty$ -- a condition that is not always true.  For example, let us take (for some real-valued state $\ket{\varphi}$) $A=\ketbra{\varphi},\ A_0 = \ketbra{\varphi}-(1/d)\id,\ A_0^2 = (1-2/d)\ketbra{\varphi}+ (1/d^2)\id$. Then we have $\tr[A_0^2]=1-1/d,\ \tr[\rho A_0]=\rho_{\varphi,\varphi}-1/d,\ \tr[\rho A_0^2]=(1-2/d)\rho_{\varphi,\varphi}+1/d^2$ and so 
\begin{align}
    {\rm Var}_{\rm Unitary} &\approx 1-1/d+1/d^2+2\rho_{\varphi,\varphi}(1-1/d)-\rho_{\varphi,\varphi}^2\\
    {\rm Var}_{\rm Real} &\approx  1/2-1/(2d)+1/d^2+2\rho_{\varphi,\varphi}(1-1/d)-\rho_{\varphi,\varphi}^2.
\end{align}
Now, if ``typically'' $\rho_{\varphi,\varphi}$ is negligible (say, $\mathcal{O}(1/d)$) then we recover the result ${\rm Var}_{\rm Unitary}/ {\rm Var}_{\rm Real} \approx 2$; on the other hand, if (say) $\rho_{\varphi,\varphi}=1$, we find ${\rm Var}_{\rm Unitary}/ {\rm Var}_{\rm Real} \approx 4/3$, a weaker improvement.

\end{document}